\documentclass{article}

\usepackage{proof,amsthm,amsfonts,amssymb,amsmath,hyperref,fontawesome,authblk,mathrsfs}
\hypersetup{colorlinks=true,linkcolor=blue,citecolor=green}

\DeclareFontFamily{U}{mathx}{\hyphenchar\font45}
\DeclareFontShape{U}{mathx}{m}{n}{
      <5> <6> <7> <8> <9> <10>
      <10.95> <12> <14.4> <17.28> <20.74> <24.88>
      mathx10
      }{}
\DeclareSymbolFont{mathx}{U}{mathx}{m}{n}
\DeclareFontSubstitution{U}{mathx}{m}{n}
\DeclareMathAccent{\widecheck}{0}{mathx}{"71}
\DeclareMathAccent{\wideparen}{0}{mathx}{"75}

\newcommand{\SL}{\mathop{/}}
\newcommand{\BS}{\mathop{\backslash}}

\newcommand{\NN}{\mathbb{N}}

\newcommand{\MALC}{\mathbf{MALC}}
\newcommand{\ACT}{\mathbf{ACT}}
\newcommand{\ACTomega}{\ACT_\omega}
\newcommand{\ACTvee}{\ACT^\vee}
\newcommand{\ACTwedge}{\ACT^\wedge}
\newcommand{\ACTbicycle}{\ACT_{\mathrm{bicycle}}}
\newcommand{\distrib}{\mathbf{D}}
\newcommand{\ACTd}{\ACT \distrib}
\newcommand{\ACTomegad}{\ACT_\omega \distrib}

\newcommand{\Lc}{\mathcal{L}}
\newcommand{\Hc}{\mathcal{H}}
\newcommand{\Cc}{\mathcal{C}}
\newcommand{\Kc}{\mathcal{K}}
\newcommand{\Mf}{\mathfrak{M}}
\newcommand{\bHc}{\overline{\Hc}}
\newcommand{\nHc}{\bHc}
\newcommand{\Gc}{\mathcal{G}}
\newcommand{\Ac}{\mathcal{A}}

\newcommand{\Z}{\mathbf{0}}
\newcommand{\mult}{\cdot}
\newcommand{\One}{\mathbf{1}}

\newcommand{\seqarr}{\vdash}
\newcommand{\cfarr}{\Rightarrow}
\newcommand{\KStar}{{}^*}

\newcommand{\inv}{\mathrm{inv}}

\newtheorem{theorem}{Theorem}
\newtheorem*{theorem*}{Theorem}
\newtheorem{lemma}{Lemma}
\newtheorem*{lemma*}{Lemma}
\newtheorem{prop}{Proposition}
\newtheorem{cor}{Corollary}

\newcommand{\CUT}{\mathrm{cut}}
\newcommand{\fp}{\mathrm{fp}}

\newcommand{\Pc}{\mathcal{P}}

\theoremstyle{definition}
\newtheorem{df}{Definition}

\newcommand{\Blank}{\mbox{\textvisiblespace}}

\begin{document}

\title{Action Logic is Undecidable}
\author{Stepan Kuznetsov}
\affil{\small Steklov Mathematical Institute of RAS}
\affil{\small and National Research University Higher School of Economics}
\maketitle

\begin{abstract}
 Action logic is the algebraic logic (inequational theory) of residuated Kleene lattices. This logic involves Kleene star, axiomatized by an induction scheme. For a stronger system which uses an $\omega$-rule instead (infinitary action logic) Buszkowski and Palka (2007) have proved $\Pi_1^0$-completeness (thus, undecidability). Decidability of action logic itself was an open question, raised by D.~Kozen in 1994. In this article, we show that it is undecidable, more precisely, $\Sigma_1^0$-complete. We also prove the same complexity results for all recursively enumerable logics between action logic and infinitary action logic; for fragments of those only one of the two lattice (additive) connectives; for action logic extended with the law of distributivity.
\end{abstract}

\section{Introduction}
\subsection{Action Lattices and Their Logics}\label{Ss:intro}

Residuated Kleene lattices (RKLs), or {\em action lattices,} are lattice structures extended simultaneously with residuals (division operations w.r.t.\ a pre-order) and iteration (Kleene star). Residuals originate in abstract algebra~\cite{Krull1924,WardDilworth1939}; then they were introduced to logic as the central component of the Lambek calculus~\cite{Lambek1958} for syntactic analysis of natural language. Nowadays residuals are viewed as a natural algebraic interpretation of implication in substructural logics~\cite{Ono1993,JipsenSurvey,GalatosRLbook,AbramskyTzevelekos2010}.

The story of {\em iteration} comes from the seminal work of S.C.~Kleene~\cite{Kleene1956}, thus its second name ``Kleene star.'' Kleene star is one of the most interesting algebraic operations in theoretical computer science. Being of inductive nature, it extends a purely propositional, algebraic logic setting with features usually found in more expressive systems, like arithmetic or higher order type theories.

The notion of residuated Kleene algebra (RKA), or action algebra, was introduced by V.~Pratt~\cite{Pratt1991}. Action algebras lack one of the lattice connectives (meet), this was added by D.~Kozen~\cite{Kozen1994}, who actually gave the definition of action lattices.

The formal definition is as follows:\footnote{We use the following notations: $\vee$ and $\wedge$ for lattice operations, $\cdot$ for product, $\BS$ and $\SL$ for residuals. In literature, notations vary: for example, $\vee$ can be replaced by $+$ (like in regular expressions), $\BS$ and $\SL$ can be written as directed implications ($\rightarrow$ and $\leftarrow$), {\em etc.} The Kleene star, however, is always denoted by $\KStar$.}
\begin{df}\label{Df:AA}
An action lattice is a structure $\langle \Ac; {\preceq}, \vee, \wedge, \Z, \mult, \One, \SL, \BS, \KStar \rangle$,
where:
\begin{enumerate}
\item $\langle \Ac; \preceq, \vee, \wedge \rangle$ is a lattice, $\Z$ is its minimal element ($\Z \preceq a$ for any $a \in \Ac$);
\item $\langle \Ac; \mult, \One \rangle$ is a monoid;
\item $\SL$ and $\BS$ are residuals of $\mult$ w.r.t.\ $\preceq$, {\em i.e.,}
$$
a \preceq c \SL b \iff a \mult b \preceq c \iff b \preceq a \BS c;
$$
\item $a^*$ is the least element $b$ such that
$\One \vee a \cdot b \preceq b$ (in other words:
$\One \preceq a^*$, $a \cdot a^* \preceq a^*$; if $\One \preceq b$ and $a \cdot b \preceq b$, then $a^* \preceq b$).
\end{enumerate}
\end{df}

The presence of residuals makes many desired properties of our algebras automatically true, so we do not need to postulate
them explicitly. These include:
\begin{itemize}
\item monotonicity of $\cdot$ w.r.t.\ $\preceq$ (shown by J.~Lambek~\cite{Lambek1958});
residuals are monotone by one argument and anti-monotone by the other one;
\item despite the asymmetry of the condition for Kleene star, the dual one also holds:
$a^*$ is also the least $b$ such that $\One \vee b \cdot a \preceq b$ (shown by Pratt~\cite{Pratt1991});
without residuals, there exist left and right Kleene algebras~\cite{Kozen1990};
\item the zero element is the annihilator w.r.t.\ $\cdot$: $\Z \cdot a = a \cdot \Z = \Z$ for any $a \in \Ac$.
\end{itemize}

Kleene~\cite{Kleene1956} informally interpreted elements of a Kleene algebra as types
of {\em events.} This interpretation gives an intuition of the Kleene algebra operations:
$a \cdot b$ means event $a$ followed by event $b$; $a \vee b$ means an event which is either
$a$ or $b$; $a^*$ is $a$ repeated several times (maybe zero\footnote{Wishing to avoid
the empty event (``nothing happened''), Kleene considered a compound connective $a^* b$, meaning
``several times $a$ followed by $b$.''}); $a \preceq b$ means that $a$ is a more specific type of events, than $b$.
Residuals also fit this paradigm. Namely, $a \BS b$ (resp., $b \SL a$) could be interpreted
as follows: this is an event which, if preceded (resp., followed) by an event of type $a$, becomes an event of type $b$.

The original setting of Kleene algebras included only three connectives: $\cdot$, $\vee$, and $\KStar$. Adding residuals and meet was motivated by the fact
that the classes of algebras in the extended setting happened to have better
properties than the original ones. 
Namely, residuated Kleene algebras form a finitely based 
variety~\cite{Pratt1991}, while Kleene algebras without residuals 
do not~\cite{Redko1964,Conway1971}. For RKLs, the
algebra of matrices over such a lattice is also an RKL, while this does not hold for RKAs (without meet)~\cite{Kozen1994}.

In computer science, the usage of Kleene algebras and their extensions is connected to reasoning about program correctness. Remarkable examples include Kleene algebras with tests~\cite{Kozen2000}, concurrent Kleene algebras~\cite{HoareCKA2011}, nominal Kleene algebras~\cite{GabbayCiancia2011,BrunetPous2016,KozenNKA2017}. Residuated Kleene algebras or lattices could also theoretically have such applications; however, there are undecidability results which make this problematic. One of such negative results is presented in this article.

Standard examples of action lattices include the algebra of languages over an alphabet and the algebra of binary relations
on a set (with $\KStar$ being the reflexive-transitive closure). Action lattices of these two classes are
{\em *-continuous} in the sense of the following definition:

\begin{df}
An action lattice is {*-continuous}, if, for any $a \in \Ac$, $a^* = \sup \{ a^n \mid n \in \omega \}$ (where
$\omega$ denotes the set of all natural numbers, including 0).
\end{df}

In the presence of residuals, we do not need the context in the definition of *-continuity
(for Kleene algebras without residuals, the condition is as follows: $b \cdot a^* \cdot c  = \sup \{ b \cdot a^n \cdot c \mid
n \in \omega \}$);  
*-continuity makes other conditions on the Kleene star (item 4 in Definition~\ref{Df:AA}) redundant.
Non-*-continuous action lattices also do exist; concrete examples are given in~\cite{Kuzn2018AiML}.

We are interested in {\em (in)equational theories,} or {\em algebraic logics,} of action lattices. Statements of such theories are of the form
 $\alpha \preceq \beta$, where $\alpha$ and $\beta$ are terms
(formulae) constructed from variables and constants $\One$ and $\Z$, using the operations of
action lattices: $\cdot$, $\BS$, $\SL$, $\vee$, $\wedge$,
$\KStar$. Statements which are true under
all interpretations of variables over arbitrary action lattices form 
{\em action logic,} denoted by $\ACT$. If we consider only
*-continuous action lattices, we get $\ACTomega$, as an extension of $\ACT$. Logics for weaker structures, which lack some of the operations, are obtained naturally as fragments as $\ACT$ or $\ACTomega$.

The motivation for considering only inequational theories is as follows. If one tries to raise the expressive power a little bit and considers {\em Horn theories,} which operate statements of the form $\alpha_1 \preceq \beta_1 \mathop{\&} \ldots \mathop{\&} \alpha_n \preceq \beta_n \Rightarrow \alpha \preceq \beta$, complexity immediately rises up to the highest possible level. Even for the language of Kleene algebras ($\cdot$, $\vee$, $\KStar$),  the Horn theory is $\Pi_1^1$-complete in the *-continuous case and $\Sigma_1^0$-complete in the general case~\cite{Kozen2002}. On the other side, in the language of residuated semigroups ($\cdot$, $\BS$, $\SL$), without Kleene star and even lattice operations, the Horn theory also happens to be $\Sigma_1^0$-complete~\cite{Buszkowski1982decision}. Thus, only for inequational theories we could expect interesting complexity results. As mentioned above, we consider these inequational as substructural propositional logics, see~\cite{GalatosRLbook}, sound and complete w.r.t.\ given algebraic semantics.  

For the *-continuous case, Buszkowski~\cite{Buszkowski2007} and Palka~\cite{Palka2007} prove undecidability and establish an exact complexity estimation of the inequational theory:
\begin{theorem}[W.~Buszkowski, E.~Palka, 2007]
 $\ACTomega$ is $\Pi_1^0$-complete.
\end{theorem}
Here the lower bound is due to Buszkowski and the upper one is due to Palka; recently A.~Das and D.~Pous~\cite{DasPous2018Action} gave another proof of the $\Pi_1^0$ upper bound for $\ACTomega$, based on non-well-founded proofs.

Notice that only the combination of residuals and meet gives this undecidability effect. The logic of residuated lattices without iteration (that is, the multiplicative-additive Lambek calculus) is decidable and PSPACE-complete~\cite{Kanovich1994PSPACE,KKS2019WoLLICcomplexity}; without $\vee$ and $\wedge$ it is NP-complete~\cite{Pentus2006}. For Kleene algebras (in the language of $\cdot$, $\vee$, $\KStar$), the logic of *-continuous Kleene algebras coincides with the logic of all Kleene algebras (but the classes of algebras do not), and this logic is also PSPACE-complete~\cite{Kozen1994IC,Krob1991}. For Kleene lattices ($\cdot$, $\vee$, $\wedge$, $\KStar$), complexity is, to the best of the author's knowledge, an open problem. However, there are decidability results on more specific classes of Kleene lattices~\cite{AndrekaMikulasNemeti2011,BrunetPous2015,Nakamura2017,DoumanePous2018} (lattices in all these classes are distributive, which is not generally true for Kleene lattices), which makes it plausible that the logic of Kleene lattices is also decidable. In contrast, lattice operations are not crucial for undecidability: the logic of *-continuous residuated monoids with iteration ($\cdot$, $\BS$, $\SL$, $\KStar$) is also $\Pi_1^0$-complete~\cite{Kuznetsov2019RSL}.

The question of decidability of $\ACT$, the logic of the whole class of action lattices, remained open, first raised by D.~Kozen in 1994~\cite{Kozen1994}. We give a negative answer: $\ACT$ is undecidable. 

This undecidability result for $\ACT$ was presented at LICS 2019 and published in its proceedings~\cite{Kuznetsov2019LICS}. This article features, besides undecidability, the following new results:
\begin{enumerate}
 \item $\Sigma_1^0$-completeness for $\ACT$ and all recursively enumerable logics in the range between $\ACT$ and $\ACTomega$;
 \item analogous results for fragments without $\vee$ and for fragments without $\wedge$;
 \item analogous results for distributive versions of $\ACT$ and its extensions up to the distributive version of $\ACTomega$.
\end{enumerate}

\subsection{Calculi: $\MALC$, $\ACTomega$, and $\ACT$}

Let us start with axiomatizing the logics introduced semantically in the previous subsection. Both $\ACT$ and $\ACTomega$ are extensions of the multiplicative-additive Lambek calculus
($\MALC$), which is the logic of residuated lattices without the Kleene star~\cite{Ono1993}.

We present $\MALC$ in the form of a Gentzen-style sequent calculus. Sequents of $\MALC$ are expressions of the form $\Gamma \seqarr \beta$, where $\beta$ is a formula (built from variables and constants $\Z$ and $\One$ using residuated lattice operations) and $\Gamma$ is a finite, possibly empty, sequence of formulae. The empty sequence is denoted by $\Lambda$. As usual, $\Gamma$ is called the {\em antecedent} and $\beta$ the {\em succedent} of the sequent. A sequent $\alpha_1, \ldots, \alpha_n \vdash \beta$
is interpreted as $\alpha_1 \cdot\ldots\cdot \alpha_n \preceq \beta$; $\Lambda \vdash\beta$ means $\One \preceq \beta$. Axioms and inference rules of $\MALC$
are as follows.
$$
\infer[(\mathrm{ax})]{\alpha \seqarr \alpha}{}
\qquad
\infer[(\Z\seqarr)]{\Gamma, \Z, \Delta \seqarr \gamma}{}
\qquad
\infer[(\One\seqarr)]{\Gamma, \One, \Delta \seqarr \gamma}{\Gamma, \Delta \seqarr \gamma}
\qquad
\infer[(\seqarr\One)]{\Lambda\seqarr\One}{}
$$
$$
\infer[(\BS\seqarr)]{\Gamma, \Pi, \alpha \BS \beta, \Delta \seqarr \gamma}{\Pi \seqarr \alpha & \Gamma, \beta, \Delta \seqarr \gamma}
\qquad
\infer[(\seqarr\BS)]{\Pi \seqarr \alpha \BS \beta}{\alpha, \Pi \seqarr \beta}
\qquad
\infer[(\cdot\seqarr)]{\Gamma, \alpha \cdot \beta, \Delta \seqarr \gamma}{\Gamma, \alpha, \beta, \Delta \seqarr \gamma}
$$
$$
\infer[(\SL\seqarr)]{\Gamma, \beta \SL \alpha, \Pi, \Delta \seqarr \gamma}{\Pi \seqarr \alpha & \Gamma, \beta, \Delta \seqarr \gamma}
\qquad
\infer[(\seqarr\SL)]{\Pi \seqarr \beta \SL \alpha}{\Pi, \alpha \seqarr \beta}
\qquad
\infer[(\seqarr\cdot)]{\Gamma, \Delta \seqarr \alpha \cdot \beta}{\Gamma \seqarr \alpha & \Delta \seqarr \beta}
$$
$$
\infer[(\wedge\seqarr)_i,\ i=1,2]{\Gamma, \alpha_1 \wedge \alpha_2, \Delta \seqarr \gamma}{\Gamma, \alpha_i, \Delta \seqarr \gamma}
\qquad
\infer[(\seqarr\wedge)]{\Pi \seqarr \alpha_1 \wedge \alpha_2}{\Pi \seqarr \alpha_1 & \Pi \seqarr \alpha_2}
$$
$$
\infer[(\vee\seqarr)]{\Gamma, \alpha_1 \vee \alpha_2, \Delta \seqarr \gamma}{\Gamma, \alpha_1, \Delta \seqarr \gamma & \Gamma, \alpha_2, \Delta \seqarr \gamma}
\qquad
\infer[(\seqarr\vee)_i,\ i=1,2]{\Pi \seqarr \alpha_1 \vee \alpha_2}{\Pi \seqarr \alpha_i}
$$

The logic for action lattices, $\ACT$ (action logic), is obtained from $\MALC$ by adding the following rules.
$$
\infer[(\KStar\seqarr)_{\mathrm{fp}}]{\alpha^* \seqarr \beta}{\Lambda\seqarr\beta & \alpha, \beta \seqarr \beta}
\qquad
\infer[(\mathrm{cut})]{\Gamma, \Pi, \Delta \seqarr \gamma}{\Pi \seqarr \alpha & \Gamma, \alpha, \Delta \seqarr \gamma}
$$
$$
\infer[(\seqarr\KStar)_0]{\Lambda\seqarr\alpha^*}{}
\qquad
\infer[(\seqarr\KStar)_{\mathrm{fp}}]{\Pi, \Delta \seqarr \alpha^*}{\Pi \seqarr \alpha & \Delta \seqarr \alpha^*}
$$

The logic for *-continuous action lattices, $\ACTomega$ (infinitary action logic), is an extension of $\MALC$ with the following rules.
$$
\infer[(\KStar\seqarr)_\omega]{\Gamma, \alpha^*, \Delta \seqarr \gamma}{\bigl(\Gamma, \alpha^n, \Delta \seqarr \gamma\bigr)_{n \in \omega}}
\qquad
\infer[(\seqarr\KStar)_n,\ n\in\omega]{\Pi_1, \ldots, \Pi_n \seqarr \alpha^*}{\Pi_1 \seqarr \alpha & \ldots & \Pi_n \seqarr \alpha}
$$

All these systems are sound and complete w.r.t.\ the corresponding classes of algebras, by Lindenbaum -- Tarski construction.

In $\ACTomega$, $(\KStar\seqarr)_\omega$ is an $\omega$-rule. The set of derivable sequents of $\ACTomega$ is defined as the smallest set including axioms and closed under rule applications. Derivations in $\ACTomega$ are possibly infinite, but well-founded trees (infinite branches forbidden).

Notice that we include cut as an official rule of the system
only in $\ACT$. Indeed, in $\ACTomega$, as shown by Palka~\cite{Palka2007}, cut
is eliminable, while for $\ACT$ no cut-free system is known. Attempts
to construct such a system were taken by P.~Jipsen~\cite{Jipsen2004} and M.~Pentus~\cite{Pentus2010iter}.
Buszkowski~\cite{Buszkowski2007} showed that in Jipsen's system cut is
not eliminable; neither it is in Pentus' systems. Constructing a cut-free system for $\ACT$ is an open problem.
Due to lack of cut elimination, we also do not know how to axiomatize elementary fragments of $\ACT$ (in restricted sublanguages) in order to guarantee conservativity.

\subsection{Some Inspiration: Circular Proofs for $\ACT$}

Before going further to proving undecidablity of $\ACT$, let us reveal some of the intuitions behind this proof. These intuitions root in non-well-founded and circular proof systems for $\ACTomega$ and $\ACT$. These systems were introduced by A.~Das and D.~Pous~\cite{DasPous2018Action}; for the identity-free version of the calculi, where empty antecedents are forbidden and Kleene star is replaced by positive iteration $\alpha^+$, they independently appear in~\cite{Kuznetsov2017WoLLIC}. 

Let us first define the non-well-founded system for $\ACTomega$, denoted by $\ACT_{\infty}$. This system arises from $\MALC$ by adding the following rules for Kleene star:
$$
\infer[(\KStar\seqarr)']{\Gamma, \alpha^*, \Delta \seqarr \gamma}
{\Gamma, \Delta \seqarr \gamma & \Gamma, \alpha, \alpha^*, \Delta \seqarr \gamma}
\qquad
\infer[(\seqarr\KStar)_0]{\Lambda \seqarr \alpha^*}{}
\qquad
\infer[(\seqarr\KStar)_{\mathrm{fp}}]
{\Pi, \Delta \seqarr \alpha^*}{\Pi \seqarr \alpha & 
\Delta \seqarr \alpha^*}
$$

The cut rule is also {\em a priori} present. All these rules are finitary. As a trade-off, we now allow non-well-founded derivations (derivations with infinite branches). The derivations should satisfy the following {\em correctness condition:} on each infinite branch of the proof, there eventually starts and continues a {\em thread} of a formula $\alpha^*$ in the antecedent, which undergoes $(\KStar\seqarr)'$ infinitely many times.

As shown by Das and Pous~\cite{DasPous2018Action}, $\ACT_\infty$ enjoys cut elimination and is equivalent  (that is, derives the same set of sequents) to $\ACTomega$.
Moreover, the {\em circular} fragment of $\ACT_\infty$ happens to be equivalent to $\ACT$. The definition of the circular fragment is as follows:
a derivation in $\ACT_\infty$ (obeying the correctness condition) is called {\em regular,} if it contains only a finite number of non-isomorphic subtrees. The term 
``circular'' comes from the following interpretation of regularity: once in an infinite derivation tree we come across a subtree which is isomorphic to the tree it contains, 
we can replace this subtree by a {\em backlink} to the root of the bigger tree (which is the same). Thus, a regular proof gets represented as a finite object, but which is 
now a graph with cycles, not a tree. 

Using cycles in derivations seems philosophically weird, reminding of {\em circuli vitiosi,} but the correctness condition guarantees that such proofs are sound. 
Unlike $\ACT_\infty$, its circular fragment does not enjoy cut elimination: if one applies the cut elimination procedure, a regular proof could become irregular. 
The circular fragment with cut, however, is equivalent to $\ACT$~\cite{DasPous2018Action}.

This circular system is not formally used in this article: we rather use a traditional formulation of $\ACT$ as presented in the previous subsection. However, it provides an 
{\em inspiration} for our undecidability proofs. Buszkowski's proof of $\Pi_1^0$-hardness of $\ACTomega$ is based on encoding the totality problem for context-free grammars,
which, in its turn, allows encoding of {\em non-halting} of Turing machines. Thus, for a Turing machine $\Mf$ and its input word $x$, one can construct a sequent which is derivable
in $\ACTomega$ if and only if $\Mf$ does not halt on $x$. Informally one can say that ``$\ACTomega$ can prove non-halting of $\Mf$ on $x$.'' Being a weaker system, $\ACT$ 
cannot prove non-halting of $\Mf$ on $x$ in all cases where it is true: otherwise, $\ACT$ would be also $\Pi_1^0$-hard, which is not the case (it is recursively enumerable).
However, in some easy cases proving non-halting in $\ACT$ is possible. We can formulate this as the following motto:
\begin{center}
 \it circular proofs for circular behaviour.
\end{center}
This roughly means that if $\Mf$ goes into a cycle on input $x$ (this is a very specific kind of non-halting), 
then the proof of $\Mf$ non-halting on $x$ also becomes circular, thus, can be carried out in $\ACT$.
Since circular behaviour of $\Mf$ on $x$ is undecidable, this leads to undecidability of $\ACT$.

We shall implement this general strategy with the following important modifications.
\begin{enumerate}
 \item Instead of considering cycling in general, we restrict ourselves to {\em trivial cycling,} where $\Mf$ just gets stuck:
 once it reaches a specific state, the rules prescribe it to stay in this state forever, neither moving nor altering the data on the tape.
 \item We have no good tools for analysis of $\ACT$ proofs (neither a cut-free system, nor reasonable semantics). Therefore, while
 we can establish the implication from circular behaviour of $\Mf$ on $x$ to derivability of the corresponding sequent in $\ACT$,
 proving the ``backwards implication,'' from derivability in $\ACT$ to circular behaviour, becomes problematic. We overcome this issue
 by using an indirect technique for proving undecidability and complexity, based on the notions of recursive inseparability (Subection~\ref{Ss:insep}) and
 effective inseparability (Section~\ref{S:sigma}).
\end{enumerate}

To conclude the introductory part, let us discuss one issue with the circular system. As one can notice, the rule $(\KStar\seqarr)'$ and
$(\seqarr\KStar)_{\fp}$ are asymmetric. On the other hand, as mentioned in Subsection~\ref{Ss:intro}, every left RKL is necessarily also a right one. Thus,
it looks plausible that adding the following ``right'' versions of these rules would not alter the set of derivable sequents.
$$
\infer[(\KStar\seqarr)'_R]{\Gamma, \alpha^*, \Delta \seqarr \gamma}{\Gamma, \Delta \seqarr \gamma & \Gamma, \alpha^*, \alpha, \beta \seqarr \gamma}
\qquad
\infer[(\seqarr\KStar)_{\fp,R}]{\Pi, \Delta \seqarr \alpha^*}{\Pi \seqarr \alpha^* & \Delta \seqarr \alpha}
$$
This is indeed true for $\ACT_\infty$---but not for its circular fragment!

In the circular fragment, {\em replacing} the ``left'' rules $(\KStar\seqarr)'$ and $(\seqarr\KStar)_{\fp}$ with the ``right'' ones,
$(\KStar\seqarr)'_R$ and $(\seqarr\KStar)_{\fp,R}$, yields the same logic, $\ACT$. However, the circular calculus including {\em both} ``left'' and
``right'' rules derives some sequents, which are not derivable in $\ACT$. An example of such a sequent is 
$(p \wedge q \wedge (p \BS q) \wedge (p \SL q))^+ \seqarr p$~\cite{Kuzn2018AiML} (here $\alpha^+ = \alpha \cdot \alpha^*$).
Thus, the circular system with both ``left'' and ``right'' rules is a natural example of an intermediate system strictly between
$\ACT$ and $\ACTomega$. Indeed, it does not coincide with $\ACT$ due to an explicit counterexample and does not coincide with $\ACTomega$, because
the latter is $\Pi_1^0$-hard, while circular systems are recursively enumerable. We denote this system, with two sets of rules for Kleene star,
by $\ACTbicycle$.

\section{Undecidability of $\ACT$}

Buszkowski~\cite{Buszkowski2007} proves $\Pi_1^0$-hardness (and thus undecidability) of the derivability problem for $\ACTomega$ by 
encoding the {\em non-halting} problem for deterministic Turing machines. In this section we extend Buszkowski's result and
prove undecidability for a range of logics. 

We consider logics in the language of $\ACT$ and $\ACTomega$ in a broad sense, just as arbitrary sets of sequents. Such a
logic will denoted by $\Lc$. The words ``$\Gamma \vdash \alpha$ is derivable $\Lc$'' mean $(\Gamma \vdash \alpha) \in \Lc$.

\begin{theorem}\label{Th:undec}
 If $\ACT \subseteq \Lc \subseteq \ACTomega$, then $\Lc$ is undecidable.
\end{theorem}

In particular, we get undecidability for $\ACTbicycle$ (introduced in the previous section), which is strictly between $\ACT$ and $\ACTomega$, and,
most importantly, for $\ACT$ itself:
\begin{cor}
  $\ACT$ is undecidable.
\end{cor}

\subsection{Encoding I: Behaviour of Turing Machines}
\label{Ss:enc1}

The proof of Theorem~\ref{Th:undec} is based on encoding behaviour of deterministic Turing machines via the totality property of context-free grammars.
Usually, in undecidability proofs one takes care about halting vs.\ non-halting of a Turing machine on a given input. This is the way 
Buszkowski's~\cite{Buszkowski2007} proof goes. In contrast, we
 distinguish three possible kinds of behaviour of a Turing machine $\Mf$ on input $x$:
\begin{enumerate}\itemsep=0pt
\item $\Mf$ halts on $x$;
\item $\Mf$ {\em trivially cycles} on $x$ (we define this notion below);
\item $\Mf$, when running on $x$, does not halt for another reason.
\end{enumerate}

 In what follows, we consider only deterministic, single-tape, single-head Turing machines.
For a Turing machine $\Mf$, let $\Sigma_0$ denote its {\em internal} alphabet (the input word is given in the {\em external} alphabet,
which is a subset of $\Sigma_0$). Let $Q$ be the set of {\em states,} with a designated initial state $q_0 \in Q$.
A {\em configuration} of $\Mf$ includes the following information: (1) the word $y = a_1 \ldots a_m$, over alphabet $\Sigma_0$, written in the internal memory;
(2) the current state $q \in Q$ of $\Mf$, and (3) which letter of $y$ is currently being observed. We encode such configurations by words over
$\Sigma_0 \cup Q$: if the machine is in state $q$ and observes letter $a_i$ of $a_1 \ldots a_m$, then this configuration is encoded as
$a_1 \ldots a_{i-1} q a_i \ldots a_m$.

The Turing machine is controlled by a finite
set of {\em rules} of the form $\langle q_1, a \rangle \to \langle q_2, b, d \rangle$, where $q_1, q_2 \in Q$, $a, b \in \Gamma$, and $d \in 
\{ L, R, N \}$. Such a rule is applied when $\Mf$ is in state $q_1$ observing letter $a$. The rule commands to replace $a$ with $b$, change
the state to $q_2$, and perform a move according to $d$. If $d = L$, move one cell left; if $d = R$, move one cell right; if $d = N$, no move is performed.
For technical reasons, we consider Turing machines with the tape growing only to the right. The left end is fixed, and if the machine
tries to go left ($d = L$) when it is already observing the leftmost cell, it halts. In contrast, if the machine is at the rightmost cell
and applies a rule with $d = R$, then the tape is extended by one cell, which is filled with a designated blank symbol $\Blank \in \Sigma_0$.
As $\Mf$ is deterministic, for each pair $\langle q_1, a \rangle$ there exists at most one rule $\langle q_1, a \rangle \to \langle q_2, b, d \rangle$.

\begin{df}
A Turing machine $\Mf$ halts on input word $x$, if it reaches a configuraton from which there is no next move
(thus, we do not distinguish ``successful'' computations from those which halt by error).
\end{df}

The notion of trivally cycling is defined as follows. Let us suppose that every Turing machine $\Mf$ includes a special {\em cycling} state
$q_c$ with rules $\langle q_c, a \rangle \to \langle q_c, a, N \rangle$ for any $a \in \Sigma_0$: once $\Mf$ reaches $q_c$, it gets stuck and 
never changes the configuration. This requirement does not restict capabilities of Turing machines, since one can just make $q_c$ unreachable.

\begin{df}
A Turing machine $\Mf$ trivially cycles on input word $x$, if $\Mf$ reaches the cycling state $q_c$ while running on $x$.
\end{df}

The notion of trivially cycling is essentially equivalent to reachability of the designated state $q_c$. For our exposition,
however, it is more convenient to consider the case of trivially cycling as a subcase of non-halting. Therefore, we force the
Turing machine to get stuck in $q_c$ forever and thus forbid halting after reaching $q_c$.

There is also a more general notion of {\em cycling} on a given input, when $\Mf$ returns to the same configuration
(and therefore runs infinitely long). For our purposes, the more restrictive notion of trivially cycling is more appropriate.

Consider the united alphabet $\Sigma = \Sigma_0 \cup Q \cup \{ \# \}$ (we suppose that $\Sigma_0 \cap Q = \varnothing$ and 
$\# \notin \Sigma_0 \cup Q$).

\begin{df}
 A {\em protocol} (computation history) of execution of $\Mf$ on input $x$ is the word $\# k_0 \# k_1 \# \ldots \# k_n \#$ over $\Sigma$, where $k_0 = q_0 x$ is the 
 (code of the) initial configuration,
 and each $k_i$ is the successor configuration of $k_{i-1}$, that is, $k_i$ is obtained from $k_{i-1}$ by applying the appropriate rule of $\Mf$.
 The protocol is a {\em halting} one, if $k_n$ has no successor. Otherwise, the protocol is {\em incomplete.}
\end{df}

Some encodings, in order to simplify proofs a bit, make configurations in a protocol  alternatingly  reversed 
($\# k_0 \# k_1^R \# k_2 \# k_3^R \# \ldots$); however, in Kozen's textbook~\cite{KozenBook} one can find
an encoding without reversions.

Let us fix $\Mf$ and its input word $x$.
Our aim is to describe all the words except the halting protocol of $\Mf$ on $x$ by a context-free grammar
$\Gc_{\Mf,x}$. Moreover, we shall provide an algorithm for constructing $\Gc_{\Mf,x}$ from $\Mf$ and $x$. 

We consider the following three classes of words, which are not the halting protocol.

\begin{enumerate}
 \itemsep=0pt
 \item Words beginning with $\#$ which {\em cannot be even a prefix} of a halting protocol. These include the following three subclasses:
\begin{itemize}
\itemsep=0pt
\item[1.1.] words which include $q_c$, where $q_c$ is the cycling state;
\item[1.2.] words which include a block between $\#$'s, which is not a code of a configuration (that is, includes zero or more than
one letters from $Q$, or the only $q \in Q$ is the rightmost letter, immediately before $\#$);
\item[1.3.] words which include a block of the form $\#k \#k'\#$, where $k'$ is not the successor of $k$;
\item[1.4.] words which start with $\# k \#$ where $k$ is not the initial configuration (that is,
$k \ne q_0 x$ if $x$ is non-empty and $k \ne q_0 \Blank$ if $x$ is empty).
\end{itemize}
\item Possibly incomplete protocols and prefixes of protocols, also beginning with $\#$. These include:
\begin{itemize}
\itemsep=0pt
\item[2.1.] words whose last symbol is not $\#$;
\item[2.2.] words of the form $u \# k \#$, where $k$ is a configuration which has a successor and $u$ is arbitrary.
\end{itemize}
\item Words not beginning with $\#$.
\end{enumerate}

Now we are ready to construct $\Gc_{\Mf,x}$, which is going to be a context-free grammar in Greibach~\cite{Greibach1965} normal form. 
First we postulate rules for a non-terminal symbol $U$ which will generate just all non-empty words: 
$$
U \cfarr a U, \qquad U \cfarr a,
$$
for all $a \in \Sigma$.

Next, construct a context-free grammar, in Greibach normal form, for all words of class~1, with the leftmost $\#$ removed. The most interesting
case here is subclass~1.3. Words of this subclass can be recognized by a non-deterministic pushdown automaton, see~\cite[Lecture~35]{KozenBook}, and it is well-known that any language recognized by a non-deterministic pushdown automaton is context-free. Removing the leftmost $\#$ from all words in this
language does not affect context-freeness.
Subclasses~1.1, 1.2, and~1.4 clearly form regular languages, and therefore are indeed context-free.
Let our context-free grammar for words of class~1, with the leftmost $\#$ removed, be in Greibach normal form and have starting symbol $Y$.

Class~2 also forms a regular language: for subclass~2.1 it is obvious, and for subclass~2.2 one just builds a finite automaton which checks
whether $k$ is a correct configuration and a rule of $\Mf$ is applicable. Thus, there is a context-free grammar for words of class~2, with the leftmost
$\#$ removed. Let this grammar also be in Greibach normal form, with non-terminals disjoint from the ones used for class~1 (and $U$). Denote the
starting symbol of this new grammar by $Z$.

Finally, we put all things together and construct $\Gc_{\Mf,x}$ adding the following rules:
\begin{align*}
& S \cfarr a && \mbox{for all $a \in \Sigma$,}\\
& S \cfarr a U && \mbox{for all $a \in \Sigma - \{ \# \}$ (this handles class~3),}\\
& S \cfarr \# Y U \\
& S \cfarr \# Z \\
& S \cfarr \# \# 
\end{align*}

Notice that $U$ appears in the production rule with $Y$, but not the one with $Z$. As mentioned above, any word
which has a {\em prefix} from class~1 is necessarily not the halting protocol. For class~2, this is not always the case.

The rule $S \cfarr \#\#$ is necessary because $S \cfarr \# Y U$ handles only words of length greater or equal than 3.
Other words of length 2 are handled by $S \cfarr aU$ or $S \cfarr \#Z$.

By construction, $\Gc_{\Mf,x}$ generates all non-empty words if and only if there exists no halting protocol, that is, $\Mf$ does not halt on $x$.

\subsection{Some Derivable Rules}\label{Ss:rules}
It will be convenient for us to consider the Kleene plus (positive iteration), defined as follows: $\psi^+ = \psi \cdot \psi^*$. In $\ACTomega$,
the Kleene plus obeys the following rules
$$
\infer[({}^+ \seqarr)_\omega]{\Gamma, \alpha^+, \Delta \seqarr \gamma}{\bigl( \Gamma, \alpha^n, \Delta \seqarr \gamma \bigr)_{n=1}^{\infty}}
\qquad
\infer[(\seqarr {}^+)_n, n \ge 1]{\Pi_1, \ldots, \Pi_n \seqarr \alpha^+}{\Pi_1 \seqarr \alpha & \ldots & \Pi_n \seqarr \alpha}
$$
(The left rule is a combination of $(\cdot\seqarr)$ and $(\KStar\seqarr)$ and the right one combines $(\seqarr\cdot)$ and $(\seqarr\KStar)$.)

We shall also consider conjunctions and disjunctions of finite sets of formulae. For $\Xi = \{ \xi_1, \ldots, \xi_n \}$ let $\bigwedge \Xi =
\xi_1 \wedge \ldots \wedge \xi_n$ and $\bigvee \Xi = \xi_1 \vee \ldots \vee \xi_n$ (the order of $\xi_i$ does not matter due to associativity and
commutativity of $\vee$ and $\wedge$). We can generalize the rules for $\wedge$ and $\vee$ in order to handle these ``big'' $\bigwedge$ and $\bigvee$:
$$
\infer[(\bigwedge\seqarr), \xi \in \Xi]{\Gamma, \bigwedge \Xi, \Delta \seqarr \gamma}{\Gamma, \xi, \Delta \seqarr \gamma}
\qquad
\infer[(\seqarr\bigwedge)]{\Pi \seqarr \bigwedge\Xi}{\bigl( \Pi \seqarr \xi \bigr)_{\xi \in \Xi}}
$$
$$
\infer[(\bigvee\seqarr)]{\Gamma, \bigvee\Xi, \Delta \seqarr \gamma}{\bigl( \Gamma, \xi, \Delta \seqarr \gamma \bigr)_{\xi \in \Xi}}
\qquad
\infer[(\seqarr\bigvee), \xi \in \Xi]{\Pi \seqarr \bigvee\Xi}{\Pi \seqarr \xi}
$$
These new ``big'' rules are obtained by applying the original ``small'' (binary) ones several times.

In order to facilitate construction of derivations in $\ACT$ and $\ACTomega$, we introduce several auxiliary rules.
These rules are going to be {\em derivable} using the rules of $\ACT$ (including cut), and thus valid in $\ACT$ and all its extensions
(including $\ACTomega$).
We start with inverting some of the rules.

\begin{lemma}
 The following rules are derivable in $\ACT$:
 $$
 \infer[(\seqarr\BS)_{\inv}]{\alpha, \Pi \seqarr \beta}{\Pi \seqarr \alpha \BS \beta}
 \qquad
 \infer[(\seqarr\SL)_{\inv}]{\Pi, \alpha \seqarr \beta}{\Pi \seqarr \beta \SL \alpha}
 \qquad
 \infer[(\cdot\seqarr)_{\inv}]{\Gamma, \alpha, \beta, \Delta \seqarr \gamma}{\Gamma, \alpha \cdot \beta, \Delta \seqarr \gamma}
 $$
 $$
 \infer[(\seqarr\wedge)_{\inv}, i =1,2]{\Pi \seqarr \alpha_i}{\Pi \seqarr \alpha_1 \wedge \alpha_2}
 \qquad
 \infer[(\vee\seqarr)_{\inv}, i=1,2]{\Gamma, \alpha_i, \Delta \seqarr \gamma}{\Gamma, \alpha_1 \vee \alpha_2, \Delta \seqarr \gamma}
 $$
 $$
 \infer[(\KStar\seqarr)_{\inv}, n\ge 0]{\Gamma, \alpha^n, \Delta \seqarr \gamma}{\Gamma, \alpha^*, \Delta \seqarr \gamma}
 \qquad
 \infer[({}^+\seqarr)_{\inv}, n\ge 1]{\Gamma, \alpha^n, \Delta \seqarr \gamma}{\Gamma, \alpha^+, \Delta \seqarr \gamma}
 $$
\end{lemma}

\begin{proof}
All these rules are established by cut, with the following sequents (respectively), which are derivable in $\ACT$:
\begin{align*}
 & \alpha, \alpha \BS \beta \seqarr \beta; \\
 & \beta \SL \alpha, \alpha \seqarr \alpha;\\
 & \alpha, \beta \seqarr \alpha \cdot \beta;\\
 & \alpha_1 \wedge \alpha_2 \seqarr \alpha_i, \quad i = 1,2;\\
 & \alpha_i \seqarr \alpha_1 \vee \alpha_2, \quad i = 1,2;\\
 & \alpha^n \seqarr \alpha^*, \quad n \ge 0;\\
 & \alpha^n \seqarr \alpha^+, \quad n \ge 1.
\end{align*}
\end{proof}
Notice that $(\KStar\seqarr)_{\inv}$ and $({}^+\seqarr)_{\inv}$, being inversions of $\omega$-rules (for Kleene star and Kleene plus
respectively), are derivable already in $\ACT$.

Consecutive applications of $(\seqarr\wedge)_\inv$ yield invertibility of the corresponding ``big'' rule, $(\seqarr\bigwedge)$; the same for $(\bigvee\seqarr)$:
$$
\infer[(\seqarr\bigwedge)_\inv, \xi \in \Xi]
{\Pi \seqarr \xi}
{\Pi \seqarr \bigwedge\Xi}
\qquad
\infer[(\bigvee\seqarr)_\inv, \xi \in \Xi]
{\Gamma, \xi, \Delta \seqarr \gamma}
{\Gamma, \bigvee\Xi, \Delta \seqarr \gamma}
$$

Next, we present a fixpoint-style rule for Kleene plus:

\begin{lemma}
 The following rule is derivable in $\ACT$:
 $$
 \infer[({}^+\seqarr)_\fp]
 {\psi^+ \seqarr \gamma}
 {\psi \seqarr \gamma & \psi, \gamma \seqarr \gamma}
 $$
\end{lemma}

\begin{proof}
 The derivation is as follows:
 $$
 \infer[(\cdot\seqarr)]{\psi^+ \seqarr \gamma}
 {\infer[(\seqarr\BS)_\inv]{\psi, \psi^* \seqarr \gamma}
 {\infer[(\KStar\seqarr)_\fp]{\psi^* \seqarr \psi \BS \gamma}
 {\infer[(\seqarr\BS)]{\Lambda \seqarr \psi \BS \gamma}
 {\psi \seqarr \gamma} &
 \infer[(\BS\seqarr)]{\psi, \psi \BS \gamma \seqarr \gamma}
 {\psi \seqarr \psi & \gamma \seqarr \gamma}}}}
 $$
\end{proof}

Finally, we establish derivability of the ``long rule'' for Kleene plus.
\begin{lemma}\label{Lm:long}
 For any natural $n$, the following ``long rule'' is admissible in $\ACT$:
 $$
 \infer{\psi^+ \vdash \gamma}{\psi \vdash \gamma & \psi^2 \vdash \gamma & \ldots & \psi^n \vdash \gamma & \psi^n, \psi^+ \vdash \gamma}
 $$
\end{lemma}

\begin{proof}
 Induction on $n$. The base case, $n=0$, is trivial (the conclusion coincides with the only premise).
 For the induction step, we 
 start deriving $\psi^+ \vdash \gamma$ by applying the ``long rule'' for $n-1$:
 $$
 \infer{\psi^+ \vdash \gamma}{\psi \vdash \gamma & \psi^2 \vdash \gamma & \ldots & \psi^{n-1} \vdash \gamma & \psi^{n-1}, \psi^+ \vdash \gamma}
 $$
 The first $n-1$ premises are given. The last one is derived as follows:
 $$
 \infer[(\CUT)]{\psi^{n-1}, \psi^+ \vdash \gamma}
 {\psi^+ \vdash \psi \vee (\psi \cdot \psi^+) & 
 \infer[(\vee\vdash)]{\psi^{n-1}, \psi \vee (\psi \cdot \psi^+) \vdash \gamma}
 {\psi^n \vdash \gamma & \infer[(\cdot\vdash)]{\psi^{n-1}, \psi \cdot \psi^+ \vdash \gamma}
 {\psi^n, \psi^+ \vdash \gamma}}}
 $$
 Here $\psi^n \vdash \gamma$ and $\psi^n, \psi^+ \vdash \gamma$ are given, and 
 $\psi^+ \vdash \psi \vee (\psi \cdot \psi^+)$ is generally true in Kleene algebra,
 thus derivable in $\ACT$.\footnote{The derivation is as follows:
 $$
 \infer[({}^+\vdash)_\fp]{\psi^+ \vdash \psi \vee (\psi \cdot \psi^+)}
 {\infer[(\vee\vdash)]
 {\psi \vdash \psi \vee (\psi \cdot \psi^+)}
 {\psi \vdash \psi} & 
 \infer[(\vee\vdash)]{\psi, \psi \vee (\psi \cdot \psi^+) \vdash 
 \psi \vee (\psi \cdot \psi^+)}
 {\infer[(\vdash\cdot)]{\psi, \psi \vee (\psi \cdot \psi^+) \vdash 
 \psi \cdot \psi^+}{\psi \vdash \psi & 
 \infer[(\vee\vdash)]{\psi \vee (\psi \cdot \psi^+) \vdash \psi^+}{\psi \vdash \psi^+ & 
 \psi \cdot \psi^+ \vdash \psi^+}}}}
 $$
 Derivations of $\psi \vdash \psi^+$ and $\psi \cdot \psi^+
 \vdash \psi^+$ are obvious.
 }
\end{proof}

Notice that disjunction ($\vee$) is not need for formulating the ``long rule,'' but is essentially used when establishing its
admissibility.

\subsection{Encoding II: from Grammars to Sequents}\label{Ss:enc2}

Let us now translate the context-free grammar $\Gc_{\Mf,x}$ into the Lambek calculus. The construction essentially resembles the translation of context-free grammars to
basic categorial grammars by Gaifman~\cite{BGS1960}. Let non-terminals of $\Gc_{\Mf,x}$ be  variables in our logics. 
For each letter $a \in \Sigma$ let 
$$
\Xi_a = \{ A \SL (B_1 \cdot \ldots \cdot B_\ell) \mid (A \cfarr a B_1 \ldots B_\ell) \text{ is a production rule of $\Gc_{\Mf,x}$} \}
$$
(in particular, for a production rule of the form $A \to a$ we have $\ell = 0$, and $A \SL (B_1 \cdot \ldots \cdot B_\ell)$ means just $A$),
$$
\varphi_a = \bigwedge \Xi_a,
$$
and
$$\psi_{\Mf,x} = \bigvee  \{ \varphi_a \mid a \in \Sigma \}.$$
Further we shall write just $\psi$ for $\psi_{\Mf,x}$, if it does not lead to confusion.

We shall need the following technical lemma about derivability in $\MALC$:
\begin{lemma}\label{Lm:theta}
Let $\Theta_1$, \ldots, $\Theta_n$ be finite sets of formulae built using only $\BS$, $\SL$, and $\cdot$, and let $\gamma$ be also built using
only $\BS$, $\SL$, $\cdot$. Then $\bigwedge\Theta_1, \ldots, \bigwedge\Theta_n \vdash \gamma$ is derivable in $\MALC$ if and only if there
exist $\theta_1 \in \Theta_1$, \ldots, $\theta_n \in \Theta_n$ such that $\theta_1, \ldots, \theta_n \vdash \gamma$ is derivable in $\MALC$.
\end{lemma}

\begin{proof}
The ``if'' part is just application of $(\bigwedge\seqarr)$. The interesting direction is ``only if.'' Consider a cut-free derivation of $\bigwedge \Theta_1, \ldots, \bigwedge \Theta_n \seqarr \gamma$ and trace the occurrences of $\bigwedge \Theta_i$ upwards from the goal sequent. After each rule application the conjunction $\bigwedge \Theta_i$ either remains intact (if it is not the active formula in this rule) or loses some of the conjuncts (actually, it either gets reduced to the rightmost conjunct, or loses this conjunct). 
The crucial observation here is that {\em the trace does not branch.} This is due to the fact that our derivation does not include $(\bigvee\seqarr)$ and $(\seqarr\bigwedge)$.
Finally, $\bigwedge \Theta_i$ gets reduced to one formula, $\theta_i  \in \Theta_i$. Then we just replace all the formulae on the trace by $\theta_i$, resulting in a valid derivation of $\theta_1, \ldots, \theta_n \seqarr \gamma$ in $\MALC$.
\end{proof}

The next three lemmas are due to Buszkowski~\cite{Buszkowski2007} and form the base for Buszkowski's proof of $\Pi_1^0$-hardness of $\ACTomega$.

\begin{lemma}\label{Lm:phi}
A word $a_1 \ldots a_n$ is generated from non-terminal $A$ in $\Gc_{\Mf,x}$ if and only if
the sequent $\varphi_{a_1}, \ldots, \varphi_{a_n} \vdash A$ is derivable in $\MALC$.~{\rm\cite{BGS1960,Buszkowski2007}}
\end{lemma}

\begin{proof}
 By Lemma~\ref{Lm:theta},  $\varphi_{a_1}, \ldots, \varphi_{a_n} \seqarr A$ is derivable if and only if there exist $\xi_1, \in \Xi_{a_1}, \ldots, \xi_n \in \Xi_{a_n}$, such that $\xi_1, \ldots, \xi_n \seqarr A$ is derivable.
 
 In order to proceed by induction, we formulate the following more general statement. Let $e_1 \ldots e_n$ be a word in the extended alphabet $\Sigma \cup N$, including both terminals and non-terminals. Then we claim that $e_1 \ldots e_n$ is derivable from $A$ in $\Gc_{\Mf,x}$ if and only if there exist such $\xi_1, \ldots, \xi_n$ that:
 \begin{enumerate}
  \item for each $i$, if $e_i \in \Sigma$, then $\xi_i \in \Xi_{e_i}$;
  \item for each $i$, if $e_i \in N$, then $\xi_i = e_i$;
  \item the sequent $\xi_1, \ldots, \xi_n \seqarr A$ is derivable.
 \end{enumerate}
 
 Both implications here are proved by induction on derivation. For the ``only if'' direction, the base case is trivial ($A \seqarr A$ is an axiom), and for the induction step let $A \cfarr a_1 B_1 \ldots B_\ell$ be the first rule applied. Then let $\xi_1 = A \SL (B_1 \cdot\ldots\cdot B_{\ell}) \in \Xi_{a_1}$, and we enjoy the following derivation:
 $$
 \infer[(\SL\seqarr)]
 {A \SL (B_1 \cdot \ldots \cdot B_{\ell}), \xi_2, \ldots, 
 \xi_n \seqarr A}
 {\xi_2, \ldots, \xi_n \seqarr B_1 \cdot \ldots \cdot B_{\ell} 
 & A \seqarr A}
 $$
 The induction hypothesis, applied to subderivations starting from $B_1$, \ldots, $B_\ell$, yields such $\xi_2, \ldots, \xi_n$ that the left premise is derivable by $(\seqarr\cdot)$.

 For the ``if'' part, we first notice that the only rules which can be applied in a cut-free derivation of $\xi_1, \ldots, \xi_n \seqarr A$ are $(\SL\seqarr)$ and $(\seqarr\cdot)$. We claim that if $\Phi \seqarr B_1 \ldots \cdot \ldots B_{\ell}$ is derivable, then $\Phi = \Phi_1, \ldots, \Phi_{\ell}$ and $\Phi_j \seqarr B_j$
 for $i = 1, \ldots, \ell$ (this is a small ``focusing lemma''). This is proved by an easy induction on derivation. Indeed, if the lowermost rule is $(\seqarr\cdot)$, we
 have $\Phi = \Phi', \Phi_{\ell}$, where $\Phi_{\ell} \seqarr B_{\ell}$ and $\Phi' \seqarr B_1 \ldots \cdot \ldots B_{\ell-1}$ are derivable. Applying the induction hypothesis to the latter, we get $\Phi' = \Phi_1, \ldots, \Phi_{\ell-1}$, with $\Phi_j \seqarr B_j$ derivable ($j = 1, \ldots, \ell-1$). If the lowermost rule is $(\SL\seqarr)$, then
 $\Phi = \Gamma, \beta \SL \alpha, \Pi, \Delta$, and $\Pi \seqarr \alpha$ and $\Gamma, \beta, \Delta \seqarr B_1 \cdot \ldots \cdot B_{\ell}$ are derivable. By induction hypothesis, $\Gamma = \Phi_1, \ldots, \Phi_{i-1}, \Phi'_i$, $\Delta = \Phi''_i, \Phi_{i+1}, \ldots, \Phi_{\ell}$, and the following sequents are derivable:
 $\Phi'_i, \beta, \Phi''_i \seqarr B_i$ and $\Phi_j \seqarr B_j$ for $j \ne i$. Applying $(\SL\seqarr)$ to the former, we get $\Phi'_i, \beta \SL \alpha, \Pi, \Phi''_i \seqarr B_i$, which is the needed $\Phi_i \seqarr B_i$ sequent.
 
 Now we proceed by induction on the total number of connectives in $\xi_1, \ldots, \xi_n \seqarr A$. This sequent is cut-free derivable, and lowermost rule in its derivation could be only $(\SL\seqarr)$ for $\xi_i = F \SL (B_1 \cdot \ldots \cdot B_{\ell})$:
 $$
 \infer{\xi_1, \ldots, \xi_{i-1}, F \SL (B_1 \cdot \ldots \cdot B_{\ell}), \xi_{i+1}, \ldots, \xi_{j}, \xi_{j+1}, \ldots, \xi_n \seqarr A}
 {\xi_{i+1}, \ldots \xi_j \seqarr B_1 \cdot \ldots \cdot B_{\ell} & 
 \xi_1, \ldots, \xi_{i-1}, F, \xi_{j+1}, \ldots, \xi_n \seqarr A}
 $$
 As shown above, derivability of $\xi_{i+1}, \ldots \xi_j \seqarr B_1 \cdot \ldots \cdot B_{\ell}$ yields derivability of
 \begin{align*}
  & \xi_{i+1}, \ldots, \xi_{k_1} \seqarr B_1; \\
  & \xi_{k_1+1}, \ldots, \xi_{k_2} \seqarr B_2; \\
  & \ldots\\
  & \xi_{k_{\ell-1}+1}, \ldots, \xi_j \seqarr B_{\ell}.
 \end{align*}
 Each of these sequents has less connectives than the original one, thus we can apply the induction hypothesis and get the following derivabilities in $\Gc_{\Mf,x}$:
 \begin{align*}
  & e_{i+1}, \ldots, e_{k_1} \mbox{ is derivable from $B_1$;}\\
  & e_{k_1+1}, \ldots, e_{k_2} \mbox{ is derivable from $B_2$;}\\
  & \ldots \\
  & e_{k_{\ell-1}+1}, \ldots, e_j \mbox{ is derivable from $B_\ell$.}
 \end{align*}
 Moreover, applying the induction hypothesis to $\xi_1, \ldots, \xi_{i-1}, F, \xi_{j+1}, \ldots, \xi_n \seqarr A$, we get
 derivability of $e_1 \ldots e_{i-1} F e_{i+1} \ldots e_n$ from $A$ in $\Gc_{\Mf,x}$ 
 Finally, since $\xi_i = F \SL (B_1 \cdot \ldots \cdot B_{\ell})$ is not a variable, $e_i = a_i$ is a terminal symbol, and
 since $\xi_i \in \Xi_i$, $F \cfarr a_i B_1 \ldots B_\ell$ is a production rule of $\Gc_{\Mf,x}$. 
  Applying this rule and the derivabilities from $B_1$, \ldots, $B_\ell$ established above, we get derivability of
  $e_1 \ldots e_n$ from $A$.
\end{proof}

\begin{lemma}\label{Lm:psi-n}
The grammar $\Gc_{\Mf,x}$ generates all words of length $n$ if and only if $\psi^n \vdash S$ is derivable in $\MALC$.
\end{lemma}

\begin{proof}
 Immediately from Lemma~\ref{Lm:phi}, by $(\bigvee\seqarr)$
 and $(\bigvee\seqarr)_\inv$.
\end{proof}

\begin{lemma}\label{Lm:ACTomega}
Turing machine $\Mf$ does not halt on input $x$ if and only if $\psi^+ \vdash S$ is derivable in $\ACTomega$.
\end{lemma}

\begin{proof}
Immediately from Lemma~\ref{Lm:psi-n}, by $({}^+\seqarr)$ and 
$({}^+\seqarr)_\inv$.
\end{proof}

This lemma yields undecidability of $\ACTomega$, since the (non-)halting problem is undecidable.
We go further and study derivability of the same sequent in $\ACT$.
Our new key lemma is as follows:
\begin{lemma}\label{Lm:ACT}
If $\Mf$ trivially cycles on $x$, then $\psi^+ \vdash S$ is derivable in $\ACT$.
\end{lemma}

\begin{proof}
 Let us first show derivability of $\psi^+ \vdash U$  in $\ACT$, that is, establish that $\ACT$ is capable
 of proving the fact that $U$ indeed generates all non-empty words. 
 Due to production rules $U \cfarr a$ and $U \cfarr aU$, we have $U \in \Xi_a$ and $U \SL U \in \Xi_a$ for any $a \in \Sigma$.
 Thus, by $(\bigwedge\vdash)$ we have $\varphi_a \vdash U$ and $\varphi_a \vdash U \SL U$,
 and by $(\bigvee\vdash)$ we get $\psi \vdash U$ and $\psi \vdash U \SL U$. Now $\psi^+ \vdash U$ is derived as follows:
  $$
 \infer[({}^+\seqarr)_\fp]
 {\psi^+ \seqarr U}
 {\psi \seqarr U & 
 \infer[(\seqarr\SL)_\inv]{\psi, U \seqarr U}
 {\psi \seqarr U \SL U}}
 $$

 Suppose that $\Mf$ trivially cycles on $x$ and consider the execution of $\Mf$ on input $x$ up to the moment when $\Mf$ enters the cycling state $q_c$.
Such an execution is unique, because $\Mf$ is deterministic.
Let the (incomplete) protocol of execution of $\Mf$ on $x$ up to the configuration with $q_c$ be of length $n$ 
($n$ is the number of letters in the protocol, not the number of configurations!).
Notice that $n\ge 3$, since this protocol includes at least $q_c$ and the $\#$'s surrounding the initial configuration.

Now we derive $\psi^+ \seqarr S$ using the ``long rule'' (Lemma~\ref{Lm:long}):
$$
\infer{\psi^+ \seqarr S}{\psi \seqarr S & \psi^2 \seqarr S & \ldots & \psi^n \seqarr S & \psi^{n}, \psi^+ \seqarr S}
$$

All its premises, except the last one, are of the form $\psi^m \seqarr S$ and are derivable by Lemma~\ref{Lm:psi-n}.
In order to derive the last premise, we first apply $(\bigvee\seqarr)$ all instances of $\psi$ in $\psi^n$. Now we
have to derive $\varphi_{a_1}, \ldots, \varphi_{a_n}, \psi^+ \seqarr S$ for any word $a_1 \dots a_n$ over $\Sigma$.

Apply cut as follows:
$$
\infer[(\mathrm{cut})]{\varphi_{a_1}, \ldots, \varphi_{a_n}, \psi^+ \seqarr S}{\psi^+ \seqarr U & \varphi_{a_1}, \ldots, \varphi_{a_n}, U \seqarr S}
$$
The left premise, $\psi^+ \vdash U$, is derivable. For the right one, notice that $a_1 \ldots a_n$ belongs to class~1 or class~3
(see Subsection~\ref{Ss:enc1}). Indeed, even if this word is a correct prefix of the (infinite) protocol of $\Mf$ on $x$, then, having length $n$,
it should include $q_c$. This makes it belong to class~1. In other cases, depending on whether $a_1 = \#$ or not, this word belongs either again to class~1, or to class~3.

If $a_1 \ldots a_n$ belongs to {\bf class~1,} then $a_1 = \#$ and $a_2 \ldots a_n$ is derivable in $\Gc_{\Mf,x}$ from $Y$. The conjunction $\varphi_{a_1}$ includes $S \SL (Y \cdot U)$, thanks to the $S \cfarr \# Y U$ production rule. Thus, $\varphi_{a_1}, \varphi_{a_2}, \ldots, \varphi_{a_n}, U \seqarr S$ can be derived as follows:
$$
\infer[(\wedge\seqarr)]{\varphi_{a_1}, \varphi_{a_2}, \ldots, \varphi_{a_n}, U \seqarr S}
{\infer[(\SL\seqarr)]{S \SL (Y \cdot U), \varphi_{a_2}, \ldots, \varphi_{a_n}, U \seqarr S}
{\infer[(\seqarr\cdot)]{\varphi_{a_2}, \ldots, \varphi_{a_n}, U \seqarr Y \cdot U}{\varphi_{a_2}, \ldots, \varphi_{a_n} \seqarr Y & U \seqarr U} & S \seqarr S}}
$$

The sequent $\varphi_{a_2}, \ldots, \varphi_{a_n} \seqarr Y$ is derivable by Lemma~\ref{Lm:phi}, since $a_2 \ldots a_n$ is derivable in $\Gc_{\Mf,x}$ from $Y$.

If $a_1 \ldots a_n$ belongs to {\bf class~3,} then $a_1 \ne \#$ and we use the $S \cfarr a_1 U$ production rule. Now $S \SL U \in \Xi_{a_1}$ and 
$U \SL U \in \Xi_{a_i}$, $i = 2, \ldots, n$ (thanks to
$U \cfarr a_i U$). Now the desired sequent $\varphi_{a_1}, \varphi_{a_2}, \ldots, \varphi_{a_n}, U \seqarr S$ can be derived, by $(\bigwedge\seqarr)$, from
$S \SL U, U \SL U, \ldots, U \SL U, U \seqarr S$. The latter is derivable in the Lambek calculus, and thus in $\ACT$.
\end{proof}

Notice that the rules with non-terminal symbol $Z$ are used only for deriving $\psi^m \seqarr S$, $m \leq n$. For longer words beginning with $\#$, we use only $Y$.

\subsection{Undecidability via Inseparability}\label{Ss:insep}

Now we are ready to prove Theorem~\ref{Th:undec}. Notice that Lemma~\ref{Lm:ACT} gives only a one-way encoding: from cycling
behaviour of $\Mf$ on input $x$ to derivability of $\psi^+ \vdash S$ in $\ACT$. If the inverse implication were also true, we
would immediately have undecidability of $\ACT$, since cycling behaviour (that is, reachability of $q_c$) is undecidable.
However, we do not have this inverse implication, and therefore use a trickier argument.

Let us introduce some notations:
\begin{align*}
 & \Cc = \{ \langle \Mf, x \rangle \mid \text{$\Mf$ trivially cycles on $x$} \}, \\
 & \Hc = \{ \langle \Mf, x \rangle \mid \text{$\Mf$ halts on $x$} \}, \\
 & \nHc = \{ \langle \Mf, x \rangle \mid \text{$\Mf$ does not halt on $x$} \}.
\end{align*}
Evidently, $\Cc \subset \nHc$, and
we shall use a folklore fact that $\Cc$ and $\Hc$ are {\em recursively inseparable:}
\begin{prop}\label{P:insep}
There exists no decidable class $\Kc$ of pairs $\langle \Mf, x \rangle$ such that
$\Cc \subseteq \Kc \subseteq \nHc$.
\end{prop}

We omit the proof of Proposition~\ref{P:insep}, since it is rather standard and, moreover, follows from a stronger
Proposition~\ref{P:CHeinsep} below. 

Next, for an arbitrary logic $\Lc$ in our language let
$$
\Kc(\Lc) = \{ \langle \Mf, x \rangle \mid \text{$\psi^+_{\Mf,x} \vdash S$ is derivable in $\Lc$} \}.
$$
If $\Lc_1 \subseteq \Lc_2$, then of course $\Kc(\Lc_1) \subseteq \Kc(\Lc_2)$.
Now Buszkowski's Lemma~\ref{Lm:ACTomega} and our Lemma~\ref{Lm:ACT} can be expressed in the following way: if $\ACT \subseteq \Lc \subseteq \ACTomega$, then
$$
\Cc \subseteq \Kc(\ACT) \subseteq \Kc(\Lc) \subseteq \Kc(\ACTomega) = \nHc.
$$
By Proposition~\ref{P:insep} $\Kc(\Lc)$ is undecidable, thus so is $\Lc$ itself. This finishes the proof of Theorem~\ref{Th:undec}.

\section{$\Sigma_1^0$-completeness of $\ACT$}\label{S:sigma}

In this section we show that our construction actually yields more than just undecidability. Namely, we prove $\Sigma_1^0$-completeness for any recursively enumerable $\Lc$ such that $\ACT \subseteq \Lc \subseteq \ACTomega$---in particular, for $\ACT$ and $\ACTbicycle$. The infinitary system $\ACTomega$ is, dually, $\Pi_1^0$-complete: the lower bound was proved by Buszkowski~\cite{Buszkowski2007}, by $\Kc(\ACTomega) = \nHc$; for the upper bound, there exist two proofs: by Palka~\cite{Palka2007} via her *-eliminating technique, and by Das and Pous~\cite{DasPous2018Action} via non-well-founded proofs. 

We follow a general road to obtain $\Sigma_1^0$-completeness results from inseparability, noticed by Speranski~\cite{Speranski2016}. The idea is to use {\em effective inseparability} instead of the usual one. The methods used come from the classics of recursive function theory.
Here we give only the
definitions and results necessary for our purposes; for a broader scope we refer to Rogers' book~\cite{Rogers}.

The theory of effective inseparability is usually developed for sets of natural numbers. Thus, we suppose that pairs $\langle \Mf, x \rangle$ (a Turing machine and its input) are encoded by natural numbers in an injective and computable way.

First, we recall the definition of a recursively enumerable (r.e.) set as the domain of a partial recursive function.
By $W_n$ we denote the domain of the partial recursive function whose program is coded by natural number $n$ ---
informally speaking, ``the $n$-th r.e. set.''

\begin{df}
 Two sets $A, B \subseteq \NN$ are called {\em effectively inseparable,} if $A \cap B = \varnothing$ and there exists
 a partial recursive function $f$ of two arguments such that if $W_u \supseteq A$, $W_v \supseteq B$ and $W_u \cap W_v = \varnothing$, then
 $f(u,v)$ is defined and $f(u,v) \notin W_u \cup W_v$.
\end{df}

The notion of effective inseparability is closely related to the notion of {\em creativity.}

\begin{df}
 A set $A \subseteq \NN$ is called {\em creative,} if $A$ is r.e.\ and there exists a partial recursive function $h$
 such that if $W_u \cap A = \varnothing$, then $h(u)$ is defined and $h(u) \notin W_u \cup A$.
\end{df}

\begin{prop}
 If $A$ and $B$ are effectively inseparable and are both r.e., then both $A$ and $B$ are creative.
\end{prop}

\begin{proof}
 Since $A$ and $B$ are r.e., we have $A = W_{u_0}$ and $B = W_{v_0}$ for some $u_0$ and $v_0$. 
 Define $h$ as follows. For any $v$ let $W_{v'} = W_{v} \cup B = W_{v} \cup W_{v_0}$. Notice that $v'$ is computable from $v$ and $v_0$.
 Next, let $h(v) = f(u_0, v')$. Let $W_v$ be an r.e.\ set disjoint with $A$. Then, since $B$ is also disjoint with $A$, so is $W_{v'}$.
 By definition of $f$, since $W_{v'} \supseteq B$, $W_{u_0} = A$, and $W_{v'} \cap W_{u_0} = \varnothing$, we see that $h(v)$ is defined and
 $h(v) \notin W_{v'} \cup W_{u_0}$. Thus, $h(v) \notin W_{v} \cup A$ (because $W_{v} \subseteq W_{v'}$). Therefore, $A$ is creative. Reasoning for $B$
 is symmetric.
\end{proof}

 For creative sets, the following Myhill's theorem establishes their $\Sigma_1^0$-comp\-le\-te\-ness.
\begin{theorem}[J.~Myhill 1955]\label{Th:Myhill}
 If $A$ is creative, than any r.e.\ set $B$ is m-reducible to $A$. In other words, any creative set is $\Sigma_1^0$-complete.~{\rm\cite[Theorem~10]{Myhill1955}}
\end{theorem}

\begin{cor}\label{Cor:insepSigma}
 If $A$ and $B$ are effectively inseparable and are both r.e., then both $A$ and $B$ are $\Sigma_1^0$-complete.~{\rm\cite[Exercise 11-14]{Rogers}}
\end{cor}

Thus, while recursive inseparability allows proving undecidability, effective inseparability is a tool for proving $\Sigma_1^0$-completeness.
In order to apply this technique to proving $\Sigma_1^0$-completeness of $\ACT$, we strengthen Proposition~\ref{P:insep} and establishes
effective inseparability of $\Cc$ and $\Hc$. 
\begin{prop}\label{P:CHeinsep}
 The sets $\Cc$ and $\Hc$ are effectively inseparable.
\end{prop}
(This is also probably a folklore fact, cf.~\cite[Exercise 7-55d]{Rogers}.)

\begin{proof}
 Recall that here we consider sets of pairs $\langle \Mf, x \rangle$ of Turing machines and their inputs, and silently suppose
 that these pairs are encoded as natural numbers. Let $W_u$ and $W_v$ be two such sets, which are both r.e., disjoint, and 
 $W_u \supseteq \Cc$, $W_v \supseteq \Hc$.
 
 The proof is a diagonalization procedure.
 Construct a Turing machine $\Mf_\eta$. Given an input $y$, $\Mf_\eta$ operates as follows:
 \begin{itemize}
  \item if $y$ is not a code of a Turing machine, then halt;
  \item if $y$ is a code of a Turing machine $\Mf$, start enumerating $W_u$ and $W_v$ in parallel, 
  waiting for $\langle \Mf, y \rangle$ to appear (since $W_u \cap W_v = \varnothing$, it could appear
  only in one enumeration); next,
  \begin{itemize}
  \item if $\langle \Mf, y \rangle \in W_u$, then halt;
  \item if $\langle \Mf, y \rangle \in W_v$, then enter the cycling state $q_c$;
  \item if neither, the machine will run forever.
  \end{itemize}
\end{itemize}

Next, let $y_\eta$ be the code of $\Mf_\eta$, and let $f(u,v) = \langle \Mf_\eta, y_\eta \rangle$. Notice that $f$ is
a computable function, since $\Mf_\eta$ ``uniformly'' depends on $u$ and $v$. Also notice that $f$ is a total function: $\Mf_\eta$
is always {\em constructed} in a finite number of steps, no matter whether it {\em runs} finitely or infinitely.

We show that $f$ is the necessary function for effective inseparability of $\Cc$ and $\Hc$. Indeed,
if $f(u,v) \in W_u$, then $\Mf_\eta$ halts on $y_\eta$, by definition of $\Mf_\eta$, that is, $\langle \Mf_\eta, y_\eta \rangle \in \Hc$.
Contradiction: $\Hc \subseteq W_v$ and $W_v \cap W_u = \varnothing$. Dually, 
if $f(u,v) \in W_v$, then $\Mf_\eta$ enters a cycling state when running on $y_\eta$, whence, $\langle \Mf_\eta, y_\eta \rangle \in \Cc$,
which is, being a subset of $W_u$, disjoint with $W_v$. Thus, $f(u,v) \notin W_u \cup W_v$, which is exactly what we need.
\end{proof}

Finally, we are ready to state and prove the main result of this section:
\begin{theorem}\label{Th:Sigma}
 If $\ACT \subseteq \Lc \subseteq \ACTomega$ and $\Lc$ is r.e., then $\Lc$ is $\Sigma_1^0$-complete. In particular,
 $\ACT$ and $\ACTbicycle$ are $\Sigma_1^0$-complete.
\end{theorem}

\begin{proof}
 Recall that $\Kc(\Lc) = \{ \langle \Mf, x \rangle \mid (\psi_{\Mf,x}^+ \vdash S) \in \Lc \}$ (by definition) and
 that $\Cc \subseteq \Kc(\ACT) \subseteq \Kc(\Lc) \subseteq \Kc(\ACTomega) = \bHc$ (Lemma~\ref{Lm:ACT} and Lemma~\ref{Lm:ACTomega}).
 By Proposition~\ref{P:CHeinsep}, $\Cc$ and $\bHc$ are effectively inseparable. Therefore, so are
 $\Kc(\Lc)$ and $\bHc$ (indeed, one can just take the same function $f$: if $W_u$ includes $\Kc(\Lc)$, it also
 includes $\Cc$). By Corollary~\ref{Cor:insepSigma}, $\Kc(\Lc)$ is $\Sigma_1^0$-complete. This implies
 $\Sigma_1^0$-hardness of $\Lc$ itself; the upper $\Sigma_1^0$ bound is given.
\end{proof}

\section{Complexity of Fragments of $\ACT$}

Our aim now is to prove undecidability and $\Sigma_1^0$-completeness results for
fragments of $\ACT$ which lack one of the additive connectives, $\wedge$ or $\vee$. 
We shall denote these fragments by $\ACTvee$ and $\ACTwedge$ respectively. Calculi for these logics are obtained from the one of $\ACT$ by removing rules for $\wedge$ and $\vee$, respectively.
The case of $\ACTvee$ is particularly interesting, since it is the inequational theory of original Pratt's action algebras~\cite{Pratt1991}.

Since we do not know whether $\ACTvee$ and $\ACTwedge$ are conservative fragments of $\ACT$,
complexity results for these systems are, formally speaking, independent ({\em i.e.,}
 neither stronger nor weaker) from the one for $\ACT$ itself. Nevertheless, by the same Lindenbaum -- Tarski completeness argument, $\ACTvee$ and $\ACTwedge$ axiomatize, respectively, action join-semilattices (that is, join-semilattices extended with residuated structure and Kleene star) and action meet-semilattices.
 
 On the other side, having a cut-free infinitary calculus, $\ACTomega$ enjoys conservativity over its fragments $\ACTomega^\vee$ and $\ACTomega^\wedge$, which are logics of *-continuous action join-semilattices and *-continuous action meet-semilattices, respectively. 

 Let us first discuss the intuition behind our construction.
In classical logic, one can get rid of $\vee$ or $\wedge$ (while keeping the other one) using de Morgan laws:
$\alpha \vee \beta \equiv \lnot (\lnot \alpha \wedge \lnot \beta)$ and $\alpha \wedge \beta  \equiv \lnot (\lnot \alpha \vee \lnot \beta)$. Unfortunately, in
$\MALC$ there is no negation. We can mimic it, however, using the following {\em pseudo-negation} construction.
Let $b$ be a fresh variable, which is going to act as the ``false'' constant; $\alpha \BS b$, denoted by $\alpha^b$, will play the r\^{o}le of
$\lnot \alpha$.

Our construction is based on the pseudo-double-negation, $\alpha^{bb} = (\alpha \BS b) \BS b$. Notice that neither $\alpha^{bb} \vdash \alpha$ 
nor $\alpha \vdash \alpha^{bb}$ is derivable in $\MALC$. The former is due to the intuitionistic nature of the system, and the latter is
due to non-commutativity. The following statement, however, allows replacing $\alpha$ with $\alpha^{bb}$:
\begin{theorem}\label{Th:DN}
 If variable $b$ does not occur in $\alpha_1, \ldots, \alpha_n, \beta$, then $\alpha_1, \ldots, \alpha_n \vdash \beta$ is derivable in $\MALC$ if and only if so is 
 $\alpha_1^{bb}, \ldots, \alpha_n^{bb} \vdash \beta^{bb}$.
\end{theorem}

Pseudo-negation in $\MALC$ enjoys one of the de Morgan laws, namely, $(\alpha \vee \beta)^b$ is equivalent to $\alpha^b \wedge \beta^b$.
This allows removing $\vee$ by replacing $(\alpha_1 \vee \ldots \vee \alpha_k)$ with $(\alpha_1 \vee \ldots \vee \alpha_k)^{bb}$, which is equivalent
to $(\alpha_1^b \wedge \ldots \wedge \alpha_k^b)^b$ and, dually, getting rid of $\wedge$ by replacing $\alpha_1 \wedge \ldots \wedge \alpha_k$
with $\alpha_1^{bb} \wedge \ldots \wedge \alpha_k^{bb}$, which is equivalent to $(\alpha_1^b \vee \ldots \vee \alpha_k^b)^b$.
Theorem~\ref{Th:DN} guarantees that adding pseudo-double-negations does not alter derivability. However, we have to be cautious, since this theorem works only for pure $\MALC$, not its extensions with Kleene star. For the latter, we still have some work to be done.

This technique of pseudo-double-negation goes back to Buszkowski~\cite{Buszkowski2007}; Theorem~\ref{Th:DN}  in its full generality appears in~\cite{KKS2019WoLLICcomplexity}. In order to make this article self-contained, we present a complete proof of Theorem~\ref{Th:DN} in the Appendix. 

Using pseudo-double-negation, we shall construct, given $\Mf$ and $x$, two formulae, \vphantom{\LARGE A}$\widecheck{\psi}$ and 
$\widehat{\psi}$. They will enjoy the same properties as 
$\psi$, but w.r.t.\ the corresponding fragments of $\ACT$ and 
$\ACTomega$, provided $S$ on the right is replaced with its pseudo-double-negation, $S^{bb}$.

Being formulae of the corresponding fragments, $\widecheck{\psi}$ and $\widehat{\psi}$ will not include
$\wedge$ and $\vee$ respectively.
For arbitrary logics $\Lc^\vee$ and $\Lc^\wedge$ (in appropriate languages), let
\begin{align*}
 & \widecheck{\Kc}(\Lc^\vee) = \{ \langle \Mf, x \rangle \mid
 {\widecheck{\psi}}^+ \vdash S^{bb} \mbox{ is derivable in $\Lc^\vee$}\},\\
 & \widehat{\Kc}(\Lc^\wedge) = \{ \langle \Mf, x \rangle \mid
 {\widehat{\psi}}^+ \vdash S^{bb} \mbox{ is derivable in $\Lc^\wedge$}\}.
\end{align*}

The desired properties of $\widecheck{\Kc}$ and $\widehat{\Kc}$ are as follows:
\begin{align*}
 & \Cc \subseteq \widecheck{\Kc}(\ACTvee) \subseteq
 \widecheck{\Kc}(\ACTomega^\vee) = \nHc,\\
 & \Cc \subseteq \widehat{\Kc}(\ACTwedge) \subseteq
 \widehat{\Kc}(\ACTomega^\wedge) = \nHc.
\end{align*}

Given such formulae $\widecheck{\psi}$ and $\widehat{\psi}$ (and an efficient method of constructing them from $\Mf$ and $x$), we 
proceed exactly as in the proofs of Theorem~\ref{Th:undec} and Theorem~\ref{Th:Sigma} and obtain the following results.

\begin{theorem}\label{Th:ACTvee}
 Any logic $\Lc^\vee$ in the language without $\wedge$ such that $\ACTvee \subseteq \Lc^\vee \subseteq \ACTomega^\vee$ is undecidable. Moreover, if such a logic is r.e., then it is $\Sigma^1_0$-complete.
\end{theorem}

\begin{theorem}\label{Th:ACTwedge}
 Any logic $\Lc^\wedge$ in the language without $\vee$ such that $\ACTwedge \subseteq \Lc^\wedge \subseteq \ACTomega^\wedge$ is undecidable. Moreover, if such a logic is r.e., then it is $\Sigma^1_0$-complete.
\end{theorem}

Let us now apply the pseudo-double-negation technique to construct $\widecheck{\psi}$ and $\widehat{\psi}$ and prove
the necessary statements about them. 
First let us construct $\widecheck{\psi}$:
\begin{align*}
& \widecheck{\varphi}_a = \left( \bigvee \{ \xi^{b} \mid \xi \in \Xi_a \} 
\right)^b;\\
& \widecheck{\psi} = \bigvee \{ \widecheck{\varphi}_a \mid a \in \Sigma \cup Q \cup \{ \# \} \}.
\end{align*}

The desired properties of $\widecheck{\psi}$ are expressed in the following lemmata. 

Notice that the proof of Theorem~\ref{Th:DN} (see Appendix) involves
analysis of cut-free derivations in $\MALC$, so it cannot be easily generalised to $\ACT$ and its fragments (for which we do not know any
cut-free calculus yet). Therefore, we are going to apply Theorem~\ref{Th:DN} only after reducing derivability questions to
sequents without Kleene star. 

\begin{lemma}\label{Lm:phivee}
 A word $a_1 \ldots a_n$ is generated from non-terminal $A$ in $\Gc_{\Mf,x}$ if and only if the sequent $\widecheck{\varphi}_{a_1}, \ldots, \widecheck{\varphi}_{a_n} \seqarr A^{bb}$ is derivable in $\MALC$ (or, equivalently, in $\ACTvee$).
\end{lemma}

This lemma is essentially due to Buszkowski~\cite{Buszkowski2007}.

\begin{proof}
 In $\MALC$, the formula $\widecheck{\varphi}_a$ is equivalent to
 $$
 \widecheck{\varphi}'_a = \bigwedge \{ \xi^{bb} \mid \xi \in \Xi_a \}.
 $$
 The sequent $\widecheck{\varphi}_{a_1}, \ldots, \widecheck{\varphi}_{a_n} \seqarr A^{bb}$ is equiderivable with
 $\widecheck{\varphi}'_{a_1}, \ldots, \widecheck{\varphi}'_{a_n} \seqarr A^{bb}$. 
 
 The latter, by Lemma~\ref{Lm:theta}, is derivable if and only if there exist such $\xi_1 \in \Xi_{a_1}$, \ldots, $\xi_n \in \Xi_{a_n}$ that  $\xi_1^{bb}, \ldots, \xi_n^{bb} \seqarr A^{bb}$  is derivable in $\MALC$. By Theorem~\ref{Th:DN}, this sequent is equiderivable with $\xi_1, \ldots, \xi_n \seqarr A$. 
 Using Lemma~\ref{Lm:theta} again we finally show that derivability of $\widecheck{\varphi}_{a_1}, \ldots, \widecheck{\varphi}_{a_n} \seqarr A^{bb}$ is equivalent to that of
 $\varphi_{a_1}, \ldots, \varphi_{a_n} \seqarr A$, and, by Lemma~\ref{Lm:phi}, equivalent to the fact that $a_1 \ldots a_n$ is generated from $A$.
\end{proof}

\begin{lemma}\label{Lm:ACTveeomega}
 $\widecheck{\Kc}(\ACTomega^\vee) = \nHc$.
\end{lemma}

This lemma is also due to Buszkowski.

\begin{proof}
 By $({}^+\seqarr)$, $(\bigvee\seqarr)$, and their inverted versions, derivability of $\widecheck{\psi}^+ \seqarr S^{bb}$ in $\ACTomega^\vee$ is equivalent to derivability of $\widecheck{\varphi}_{a_1}, \ldots, \widecheck{\varphi}_{a_n} \seqarr S^{bb}$ in $\MALC$ for any non-empty word $a_1 \ldots a_n$ over $\Sigma$.
 By Lemma~\ref{Lm:phivee}, this is equivalent to the fact that any non-empty word $a_1 \ldots a_n$ is generated from $S$ in $\Gc_{\Mf,x}$. In turn, this is the case if and only if $\Mf$ does not halt on $x$. 
\end{proof}

\begin{lemma}\label{Lm:ACTvee}
 $\Cc \subseteq  \widecheck{\Kc}(\ACTvee)$.
\end{lemma}

\begin{proof}
We reproduce the proof of Lemma~\ref{Lm:ACT} adding
double negations where needed. First we derive 
$\widecheck{\psi}^+ \seqarr U^{bb}$:
$$
\infer[({}^+\seqarr)_\fp]{\widecheck{\psi}^+ \seqarr U^{bb}}
{\widecheck{\psi} \seqarr U^{bb} & 
\infer[(\CUT)]{\widecheck{\psi}, U^{bb} 
\seqarr U^{bb}}{\widecheck{\psi} \seqarr (U \SL U)^{bb} &
(U \SL U)^{bb}, U^{bb} \seqarr U^{bb}}}
$$
Here $(U \SL U)^{bb}, U^{bb} \seqarr U^{bb}$ is obtained from
$U \SL U, U \seqarr U$ by Theorem~\ref{Th:DN}; the latter is obviously derivable. For $\widecheck{\psi} \seqarr U^{bb}$ and $\widecheck{\psi} \seqarr (U \SL U)^{bb}$ we use the following trick. By $(\bigvee\seqarr)$, it is sufficient to prove $\widecheck{\varphi}_a \seqarr U^{bb}$ and 
$\widecheck{\varphi}_a \seqarr (U \SL U)^{bb}$ for every
$a \in \Sigma$. In $\MALC$, we can replace $\widecheck{\varphi}_a$ with equivalent $\widecheck{\varphi}'_a$. The latter is a conjunction which includes $U^{bb}$ and $(U \SL U)^{bb}$
(thanks to the corresponding production rules of $\Gc_{\Mf,x}$).
Since $\MALC$ enjoys cut elimination, these derivations can be performed without using $\wedge$ (which appears in $\widecheck{\varphi}'_a$).

 Now let $\Mf$ trivially cycle on $x$. Notice that the ``long rule'' is derived without using $\wedge$ and is therefore valid in $\ACTvee$. Apply it:
 $$
 \infer{\widecheck{\psi}^+ \vdash S^{bb}}{\widecheck{\psi} \vdash S^{bb} & \ldots & \widecheck{\psi}^n \vdash S^{bb} & \infer[(\CUT)]{\widecheck{\psi}^n, \widecheck{\psi}^+ \vdash S^{bb}}{\widecheck{\psi}^+ \seqarr U^{bb} & 
 \widecheck{\psi}^n, U^{bb} \seqarr S^{bb}}}
 $$
 Here, as in Lemma~\ref{Lm:ACT}, $n$ is the length of the protocol up to reaching $q_c$.
 
 The first $n$ premises are derived exactly as in the proof
 of Lemma~\ref{Lm:ACTveeomega}, applying $(\bigvee\seqarr)$ to sequents $\widecheck{\varphi}_{a_1}, \ldots, \widecheck{\varphi}_{a_n} \seqarr S^{bb}$, which are derivable in $\MALC$.
 
 The rightmost premise,
 $\widecheck{\psi}^n, U^{bb} \seqarr S^{bb}$, is derived by
 $(\bigvee\seqarr)$ from all sequents 
 $\widecheck{\varphi}_{a_1}, \ldots, \widecheck{\varphi}_{a_n}, U^{bb} \seqarr S^{bb}$ or, equivalently in $\MALC$, 
 $\widecheck{\varphi}'_{a_1}, \ldots, \widecheck{\varphi}'_{a_n}, U^{bb} \seqarr S^{bb}$.
 
 As in the proof of Lemma~\ref{Lm:ACT}, the word $a_1\ldots a_n$ belongs to class~1 or class~3, depending on whether $a_1 = \#$.
 
 If $a_1 = \#$ and $a_1 \ldots a_n$ is a word of class~1, then the conjunction $\widecheck{\varphi}'_{a_1}$ contains $(S \SL (Y \cdot U))^{bb}$, and we proceed by $(\bigwedge\seqarr)$ as follows:
 $$
 \infer[(\wedge\seqarr)]
 {\widecheck{\varphi}'_{a_1}, \ldots, \widecheck{\varphi}'_{a_n}, U^{bb} \seqarr S^{bb}}
 {\infer[(\CUT)]{(S \SL (Y \cdot U))^{bb}, \widecheck{\varphi}'_{a_2}, \ldots,
 \widecheck{\varphi}'_{a_n}, U^{bb} \seqarr S^{bb}}
 {\widecheck{\varphi}'_{a_2}, \ldots,
 \widecheck{\varphi}'_{a_n} \seqarr Y^{bb} & 
 (S \SL (Y \cdot U))^{bb}, Y^{bb} \seqarr S^{bb}}}
 $$
 Here the left premise is derivable by Lemma~\ref{Lm:phivee} (because $a_2 \ldots a_n$ is generated from $Y$),
 and the right one is obtained from $S \SL (Y \cdot U), Y, U \seqarr S$ by Theorem~\ref{Th:DN}.
 
 If $a_1 \ne \#$ and $a_1 \ldots a_n$ is a word of class~3, then the conjunction $\widecheck{\varphi}'_{a_1}$ contains $(S \SL U)^{bb}$ and conjunctions $\widecheck{\varphi}'_{a_i}$, for $i = 2, \ldots, n$, contain $(U \SL U)^{bb}$. Now 
 $\widecheck{\varphi}'_{a_1}, \widecheck{\varphi}'_{a_2}, \ldots, \widecheck{\varphi}'_{a_n}, U^{bb} \seqarr S^{bb}$ is derived using $(\bigwedge\seqarr)$ from the sequent $(S\SL U)^{bb}, 
 (U \SL U)^{bb}, \ldots, (U \SL U)^{bb}, U^{bb} \seqarr  S^{bb}$. The latter is derivable by Theorem~\ref{Th:DN}.
 
 In both cases, we have derived $\widecheck{\varphi}_{a_1}, \widecheck{\varphi}_{a_2}, \ldots, \widecheck{\varphi}_{a_n}, U^{bb} \seqarr S^{bb}$ via a detour using $\wedge$ (which appears in $\widecheck{\varphi}'_{a_i}$). The derivations, however, are performed in pure $\MALC$ (no Kleene star), which enjoys cut elimination. After performing cut elimination, we obtain a derivation without $\wedge$, which is legal in $\ACTvee$. 
 \end{proof}

Lemmata~\ref{Lm:ACTveeomega} and~\ref{Lm:ACTvee}, via Proposition~\ref{P:insep}, Proposition~\ref{P:CHeinsep}, and Corollary~\ref{Cor:insepSigma}, yield Theorem~\ref{Th:ACTvee}. 

Now let us do the job for $\ACTwedge$.
In order to prove Theorem~\ref{Th:ACTwedge}, we construct $\widehat{\psi}$, using pseudo-double-negation in a slightly different way:
$$
\widehat{\psi} = \left( \bigwedge \{ \varphi_a^b \mid a \in \Sigma \} \right)^b,
$$
where $\varphi_a$ is defined in the standard way, like in the definition of $\psi$: $\varphi_a = \bigwedge \Xi_a$.

In $\MALC$, $\widehat{\psi}$ is equivalent to
$$
\widehat{\psi}' = \bigvee \{ \varphi_a^{bb} \mid a \in \Sigma \}.
$$
Thus, in our derivations, once we reach a point where there are no more occurrences of Kleene star, we can replace $\widehat{\psi}$ by $\widehat{\psi}'$, decompose $\bigvee$ and directly apply Theorem~\ref{Th:DN}. The instances of $\vee$ will be removed by cut elimination in $\MALC$. For the Kleene star, however, we still have to do some work.

First, we have to reestablish the ``long rule,'' since its old derivation (see Lemma~\ref{Lm:long}) essentially uses $\vee$, which is now unavailable.
\begin{lemma}\label{Lm:longwedge}
 The ``long rule'' is derivable in $\ACTwedge$.
\end{lemma}

\begin{proof}
 It is sufficient to show that one can derive $\psi^{n-1}, \psi^+ \seqarr \gamma$ from $\psi^n \seqarr \gamma$ and $\psi^n, \psi^+ \seqarr \gamma$, using $\wedge$ instead of $\vee$. Everything else is already done in the proof of Lemma~\ref{Lm:long}.
 
 First we derive $(\gamma \SL \psi) \wedge (\gamma \SL (\psi \cdot \psi^+)) \seqarr \gamma \SL \psi^+$ in $\ACTwedge$:
 $$
 \infer{(\gamma \SL \psi) \wedge (\gamma \SL \psi\psi^+) \seqarr \gamma \SL \psi^+}
{\infer{(\gamma \SL \psi) \wedge (\gamma \SL \psi\psi^+), \psi^+ \seqarr \gamma}
{\infer{(\gamma \SL \psi) \wedge (\gamma \SL \psi\psi^+), \psi, \psi^* \seqarr \gamma}
{\infer{\psi^* \seqarr (((\gamma \SL \psi) \wedge (\gamma \SL \psi\psi^+)) \cdot \psi) \BS \gamma}
{\infer{\Lambda \seqarr (((\gamma \SL \psi) \wedge (\gamma \SL \psi\psi^+)) \cdot \psi) \BS \gamma}
{\infer{((\gamma \SL \psi) \wedge (\gamma \SL \psi\psi^+)), \psi \seqarr \gamma}{\gamma \SL \psi, \psi \seqarr \gamma}} & 
\infer{\psi, \psi^* \seqarr (((\gamma \SL \psi) \wedge (\gamma \SL \psi\psi^+)) \cdot \psi) \BS \gamma}
{\infer{(\gamma \SL \psi) \wedge (\gamma \SL \psi\psi^+), \psi, \psi, \psi^* \seqarr \gamma}
{\infer{\gamma \SL \psi\psi^+, \psi, \psi, \psi^* \seqarr \gamma}{\psi, \psi, \psi^* \seqarr \psi\psi^+}}}}}}}
$$
(In order to make notations shorter, we write $\psi\psi^+$ instead of $(\psi \cdot \psi^+)$.)

Now we proceed as follows:
$$
\infer{\psi^{n-1}, \psi^+ \seqarr \gamma}
{\infer{\psi^{n-1} \seqarr \gamma \SL \psi^+}
{\infer{\psi^{n-1} \seqarr (\gamma \SL \psi) \wedge (\gamma \SL \psi \psi^+)}
{\infer{\psi^{n-1} \seqarr \gamma \SL \psi}
{\psi^n \seqarr \gamma} & \infer{\psi^{n-1} \seqarr \gamma \SL \psi \psi^+}{\infer{\psi^{n-1}, \psi \psi^+ \seqarr \gamma}{\psi^n,\psi^+ \seqarr \gamma}}} & 
(\gamma \SL \psi) \wedge (\gamma \SL \psi  \psi^+) \seqarr \gamma \SL \psi^+}}
$$
\end{proof}

Now we are ready to prove that $\Cc \subseteq \widehat{\Kc}(\ACTwedge)$ and $\widehat{\Kc}(\ACTomega^\wedge) = \nHc$.

\begin{lemma}\label{Lm:ACTomegawedge}
 $\widehat{\Kc}(\ACTomega^\wedge) = \nHc$.
\end{lemma}

\begin{proof}
 We prove that $\widehat{\Kc}(\ACTomega^\wedge) = 
 \Kc(\ACTomega)$ and then use Lemma~\ref{Lm:ACTomega}. First, by conservativity, we can replace $\ACTomega^\wedge$ by the full system $\ACTomega$, this will not affect derivability of $\widehat{\psi}^+ \vdash S^{bb}$.

 The sequent $\psi^+ \vdash S$ is derivable in $\ACTomega$ if and only if so are sequents $\varphi_{a_1}, \ldots, \varphi_{a_n} \vdash S$ for all non-empty words $a_1, \ldots, a_n$. On the other side, $\widehat{\psi}^+ \vdash S$ is equivalent to $\widehat{\psi}'^+ \vdash S$, which is derivable if and only if so are all sequents of the form $\varphi_{a_1}^{bb}, \ldots, \varphi_{a_n}^{bb} \vdash S^{bb}$. Equiderivability of $\varphi_{a_1}, \ldots, \varphi_{a_n} \vdash S$ and $\varphi_{a_1}^{bb}, \ldots, \varphi_{a_n}^{bb} \vdash S^{bb}$ is due to Theorem~\ref{Th:DN}.
  \end{proof}

\begin{lemma}\label{Lm:ACTwedge}
$\Cc \subseteq \widehat{\Kc}(\ACTwedge)$.
\end{lemma}

\begin{proof}
 We start, again, with deriving $\widehat{\psi}^+ \seqarr U^{bb}$:
 $$
 \infer[({}^+\seqarr)]{\widehat{\psi}^+ \seqarr U^{bb}}
 {\widehat{\psi} \seqarr U^{bb} & 
 \infer[(\CUT)]{\widehat{\psi}, U^{bb} \seqarr U^{bb}}
 {\widehat{\psi} \seqarr (U \SL U)^{bb} & 
 (U \SL U)^{bb}, U^{bb} \seqarr U^{bb}}}
 $$
 The sequents $\widehat{\psi} \seqarr U^{bb}$ and 
 $\widehat{\psi} \seqarr (U \SL U)^{bb}$ here do not contain ${}^*$. In $\MALC$, we equivalently replace $\widehat{\psi}$ with
 $\widehat{\psi}'$ and use $(\bigvee\seqarr)$. Now we have to derive $\varphi_{a}^{bb} \seqarr U^{bb}$ and $\varphi_a^{bb} \seqarr (U \SL U)^{bb}$ for any $a \in \Sigma$. These sequents follow from $\varphi_a \seqarr U$ and $\varphi_a \seqarr U \SL U$ by Theorem~\ref{Th:DN}. After cut elimination (in $\MALC$), these derivations become $\vee$-free. As for $(U \SL U)^{bb}, U^{bb} \seqarr U^{bb}$, it is obtained from $U \SL U, U \seqarr U$ by Theorem~\ref{Th:DN}.
 
 Now let $\Mf$ trivially cycle on $x$, and let $n$ be the length of its protocol up to reaching the cycling state $q_c$. We derive $\widehat{\psi}^+ \seqarr S^{bb}$ using the ``long rule'' (Lemma~\ref{Lm:longwedge}):
 $$
 \infer{\widehat{\psi}^+ \seqarr S^{bb}}
 {\widehat{\psi} \seqarr S^{bb} & \ldots & 
 \widehat{\psi}^n \seqarr S^{bb} & 
 \infer{\widehat{\psi}^n, \widehat{\psi}^+ \seqarr S^{bb}}
 {\widehat{\psi}^+ \seqarr U^{bb} & 
 \widehat{\psi}^n, U^{bb} \seqarr S^{bb}}}
 $$
 
 Again, we have to derive $\widehat{\psi}^n, U^{bb} \seqarr S^{bb}$ in pure $\MALC$. Replace $\widehat{\psi}$ with
 $\widehat{\psi}'$ and use $(\bigvee\seqarr)$. Now we have to derive $\varphi_{a_1}^{bb}, \ldots, \varphi_{a_n}^{bb}, U^{bb} \seqarr S^{bb}$ for any word $a_1 \ldots a_n$ over $\Sigma$. 
 
 On the other hand (see proof of Lemma~\ref{Lm:ACT}), we know  that $\varphi_{a_1}, \ldots, \varphi_{a_n}, U \seqarr S$ is derivable. We conclude by applying Theorem~\ref{Th:DN} and noticing that occurrences of $\vee$ are removed from our $\MALC$-derivation by cut elimination.
\end{proof}

Now we again use our standard machinery (Propositions~\ref{P:insep} and~\ref{P:CHeinsep} and Corollary~\ref{Cor:insepSigma}) and establish Theorem~\ref{Th:ACTwedge}.

\section{Action Logic with Distributivity}

As usual in substructural logics~\cite{OnoKomori1985}, the distributivity law for $\vee$ and $\wedge$ is not
derivable in $\MALC$ (and therefore in $\ACT$ and $\ACTomega$). 
In private communication with the author, Igor Sedl\'{a}r raised a question whether the undecidability and complexity
results presented above keep valid if one adds distributivity as an extra axiom to $\ACT$ and its extensions up to
$\ACTomega$. In this section we show that they do.

Let $\distrib$ denote the distributivity principle in the form 
$$(\alpha \vee \beta) \wedge (\alpha \vee \gamma) \vdash
\alpha \vee (\beta \wedge \gamma)$$ 
(the converse is derivable already in $\MALC$), and let $\ACTd$ stand for $\ACT$ extended
with $\distrib$ as an extra axiom, ditto for $\ACT_\omega \distrib$.

Our key observation is that adding distributivity to $\ACTomega$ does not affect derivability for sequents we need for
proving our complexity results.
In order to prove this, we follow Buszkowski~\cite{BuszkoRelMiCS} and use a specific class of distributive residuated
Kleene lattices, namely, lattices of binary relations (R-lattices for short), in their unrelativized, ``square'' version.

\begin{df}
An R-lattice is an algebraic structure build over $\Ac = \Pc(U \times U)$ for a non-empty set $U$, i.e., the set of all binary relations on $U$.
The pre-order and operations are defined as follows:
\begin{enumerate}
 \item the lattice structure is set-theoretic:
 $\preceq$ is the subset relation, $\vee$ is union, $\wedge$ is intersection; $\Z$ is the empty set;
 \item product is relation composition:
 $$R \cdot S = \{ \langle x,z \rangle \in U \times U \mid
 (\exists z \in U) \, \langle x,y \rangle \in R \mbox{ and }
 \langle y,z \rangle \in S \};
 $$
 the multiplicative unit is the diagonal relation:
 $\One = \{ \langle x,x \rangle \mid x \in U \}$;
 \item residuals are defined in the only way to satisfy the conditions of Definition~\ref{Df:AA}:
 \begin{align*}
  & R \BS S = \{ \langle y,z \rangle \in U\times U \mid
  R \cdot \{ \langle y,z \rangle \} \subseteq S \}, \\
  & S \SL R = \{ \langle x,y \rangle \in U \times U \mid
  \{ \langle x,y \rangle \} \cdot R \subseteq S \};
 \end{align*}
 \item Kleene star is the reflexive-transitive closure:
 $$
 R^* = \bigcup_{n=0}^{\infty} R^n
 $$
 (here $R^n = \underbrace{R \cdot \ldots \cdot R}_{\text{$n$ times}}$ and $R^0 = \One$).
\end{enumerate}

\end{df}

Since R-lattices are a particular case of *-continuous RKLs and they are distributive, $\ACTomega\distrib$ is sound w.r.t.\ interpretations on R-lattices.
Completeness of $\ACT_\omega \distrib$, without the $\Z$ and $\One$ constants\footnote{These constants also cause problems with completeness.} 
is an open problem. The system $\ACTomega$ and its *-free fragment $\MALC$ are incomplete w.r.t.\ R-lattices due to the lack of the distributivity law.
The fragment without $^*$, $\vee$, $\Z$, and $\One$, however, which we denote by $\MALC^\wedge$, {\em is} complete w.r.t.\ R-lattices.
This was essentially proved by Andr\'{e}ka and Mikul\'{a}s~\cite{AndrekaMikulas}: their Theorem~3.1 formally does not take care for $\wedge$, but can be easily extended to $\MALC^\wedge$. 
\begin{theorem}[Andr\'{e}ka and Mikul\'{a}s, 1994]\label{Th:Andreka}
 A sequent without $^*$, $\vee$, $\Z$, and $\One$ is true on all R-lattices if and only if it is derivable in $\MALC^\wedge$.
\end{theorem}
In particular, $\MALC^\wedge$ is a conservative fragment of {\em both} $\ACTomega$ and $\ACT_\omega \distrib$, and also
of both $\ACT$ and $\ACTd$.

Now we are ready to prove the main lemma of this section.
\begin{lemma}\label{Lm:distrconserv}
 If $\psi^+ \vdash S$, where $\psi$ is the formula constructed from a Turing machine $\Mf$ and its input word $x$ as
 described in Subsection~\ref{Ss:enc2}, is derivable in $\ACT_\omega \distrib$, then it is also derivable in $\ACT_\omega$ without
 the distributivity axiom. 
\end{lemma}

Notice that the derivable rules presented in Subsection~\ref{Ss:rules} are still valid in $\ACT_\omega\distrib$ and $\ACT\distrib$, so we can freely use them.

\begin{proof}
  Apply $({}^+\seqarr)_{\inv}$ and $(\bigvee\seqarr)_\inv$ to $\psi^+ \seqarr S$. 
 This yields derivability of  $\varphi_{a_1}, 
 \ldots, \varphi_{a_n} \vdash S$  in $\ACTomega \distrib$ for any word $a_1 \ldots a_n$ over $\Sigma$.
 This sequent does not include ${}^*$ and $\vee$, thus, via R-completeness, it is also derivable in $\ACTomega$. Now
 $(\bigvee\seqarr)$ and $({}^+\seqarr)$ give derivability of $\psi^+ \seqarr S$.
 \end{proof}
 
 This lemma also holds $\widehat{\psi}^+ \seqarr S^{bb}$ and $\widecheck{\psi}^+ \seqarr S^{bb}$. However, we cannot formulate the corresponding fragments of $\ACTd$, since the distributivity law includes both $\vee$ and $\wedge$.
 
 Now we are ready to prove a distributive analog of Theorem~\ref{Th:undec} and Theorem~\ref{Th:Sigma}.

 \begin{theorem}\label{Th:ACTdistr}
 Any logic $\Lc$ such that $\ACT \subseteq \Lc \subseteq \ACTomegad$ is undecidable. Moreover, if it is r.e., then it is $\Sigma_1^0$-complete. 
 \end{theorem}

\begin{proof}
 Since $\Cc \subseteq \Kc(\ACT)$ (Lemma~\ref{Lm:ACT}) and 
 $\Kc(\ACTomegad) = \Kc(\ACTomega) = \nHc$ (Lemma~\ref{Lm:distrconserv} and Lemma~\ref{Lm:ACTomega}), we have $$\Cc \subseteq \Kc(\Lc) \subseteq \nHc.$$ By Proposition~\ref{P:insep}, $\Kc(\Lc)$, and therefore $\Lc$ itself, is undecidable. If $\Lc$ is r.e., then Proposition~\ref{P:CHeinsep} and Corollary~\ref{Cor:insepSigma} yield its $\Sigma_1^0$-completeness.
\end{proof}

\begin{cor}
 $\ACTd$ is $\Sigma_1^0$-complete.
\end{cor}

Notice that Theorem~\ref{Th:Sigma} is insufficient to prove undecidability of $\ACTd$, since $\ACTd \not\subseteq \ACTomega$. On the other hand, the modifications performed here are solely on the *-continuous side, and actually easily follow from Buszkowski's results on R-lattices~\cite{BuszkoRelMiCS}.

 \subsection*{Future Work}
 To conclude, we formulate three problems in the field which are still open.
 
 \begin{enumerate}
  \item For $\ACT$, there is no good (cut-free) sequent calculus known. This is quite annoying, since we do not actually have any good tool for analysis of derivability in $\ACT$. Indeed, no reasonable semantics for $\ACT$ is known either. Natural classes of models are *-continuous, thus, they are models for $\ACTomega$ and $\ACT$ is definitely incomplete w.r.t.\ them. If one needs to distinguish $\ACT$ from $\ACTomega$, one has to invent {\em ad hoc} algebraic constructions~\cite{Kuzn2018AiML}.
  
  \item Decidability for the logic of Kleene lattices without distributivity is still an open question, as well as constructing a finitary (for example, circular) cut-free calculus.
  
  \item Finally, we have a question of complexity of $\ACT$ without both $\vee$ and $\wedge$, that is, the Lambek calculus with inductively axiomatized Kleene star. Compare with $\Pi_1^0$-completeness for the corresponding fragment of $\ACTomega$  \cite{Kuznetsov2019RSL}. If one tries to use the strategy of this article, the following problem arises: we need $\vee$ or $\wedge$ to derive the ``long rule,'' which is necessary for our encoding. As a first step, one could try to prove undecidability of $\ACT$ without $\vee$ and $\wedge$, but with the ``long rule'' explicitly added to the calculus.
 \end{enumerate}

\subsection*{Acknowledgments}

The author is grateful to Lev Beklemishev, Anupam Das, Max Kanovich, Fedor Pakhomov, Andre Scedrov, Igor Sedl\'ar, Daniyar Shamkanov, and Stanislav Speranski for fruitful
discussions. Being a Young Russian Mathematics award winner, the author thanks its jury and sponsors for this high honour.

\bibliographystyle{abbrv}
\bibliography{ACT.bib}

\appendix 
\section*{Appendix. Proof of Theorem~\ref{Th:DN}}

In this Appendix we present the proof of Theorem~\ref{Th:DN} (pseudo-double-negation theorem). This theorem actually has nothing to do with the Kleene star, it is a property of $\MALC$. Moreover, it is are more like folklore, not our new contribution. For these reasons, and also in order to avoid overloading the article with extra technicalities, we refrained from putting the proof in the article itself. We rather present it here in the Appendix in order to keep the article logically self-contained. 

In what follows, ``derivable'' always means ``derivable in $\MALC$,'' no other calculi considered beyond this point.

Theorem~\ref{Th:DN} states that adding pseudo-double-negation does not alter derivability. First we prove several auxiliary statements.

\begin{lemma}\label{Lm:aux1}
If $\alpha_1$, \ldots, $\alpha_\ell$ do not include $b$, then $\alpha_1, \ldots, \alpha_\ell \seqarr b$ is not derivable.
\end{lemma}

\begin{proof}
No left occurrence of $b$ available for the $b \seqarr b$ axiom in a cut-free derivation.
\end{proof}

\begin{lemma}\label{Lm:aux2}
If $\alpha_1, \ldots, \alpha_\ell, \beta_1 \BS b, \ldots, \beta_k \BS b, b \seqarr b$ is derivable and $\alpha_1, \ldots, \alpha_\ell$ do not include $b$, then $k=\ell=0$.
\end{lemma}

\begin{proof}
Induction on cut-free derivation. 

Suppose $k \ne 0$ or $\ell \ne 0$. Then the lowermost rule either decomposes one of $\alpha_i$, or is $(\BS \seqarr)$ decomposing one of $\beta_j \BS b$. 

Let us call a premise of this rule {\em main}, if it contains the succedent $b$ of the goal sequent. For $(\BS\seqarr)$ and $(\SL\seqarr)$, the main premise is the right one; for $(\vee\seqarr)$, both premises are main, and all other left rules of $\MALC$ have only one premise.

Now consider two cases.

{\em Case 1:} the rightmost $b$ in the antecedent is kept in the main premise. Then, by induction hypothesis, this main premise is $b \seqarr b$. But then the rule cannot be applied. Contradiction.

{\em Case 2:} the rightmost $b$ goes to the non-main (left) premise. This means that the rule should be $(\SL\seqarr)$, since there is no $\BS$ to the right of this $b$:
$$
\infer{\alpha_1, \ldots, \alpha'_i \SL \alpha''_i, 
\alpha_{i+1}, \ldots, \alpha_\ell, \beta_1 \BS b, \ldots, \beta_k \BS b, b \seqarr b}
{\alpha_{i+1}, \ldots, \alpha_\ell, \beta_1 \BS b, \ldots, \beta_k \BS b, b \seqarr \alpha''_i & 
\alpha_1, \ldots, \alpha'_i \seqarr b}
$$
Derivability of the right (main) premise here contradicts Lemma~\ref{Lm:aux1}.
\end{proof}

\begin{lemma}\label{Lm:aux3}
 If $\alpha_1, \ldots, \alpha_\ell, \beta_1 \BS b, \ldots, \beta_k \BS b \seqarr b$ is derivable and $\alpha_1$, \ldots, $\alpha_\ell$ do not include $b$, then $\alpha_1, \ldots, \alpha_\ell, \beta_1 \BS b, \ldots, \beta_{k-1} \BS b \seqarr \beta_k$ is also derivable. 
\end{lemma}

\begin{proof}
Let $\Gamma = \alpha_1, \ldots, \alpha_\ell, \beta_1 \BS b, \ldots, \beta_{k-1} \BS b$.  Proceed, again, by induction on the (cut-free) derivation of $\Gamma, \beta_k \BS b \seqarr b$. Consider the lowermost rule. We have three possibilities.

{\em Case 1:} the rule operates inside $\Gamma$, and its main premise is $\widetilde{\Gamma}, \beta_k \BS b \seqarr b$. Then by induction hypothesis we get derivability of $\widetilde{\Gamma} \seqarr \beta_k$, and applying the same rule yields $\Gamma \seqarr \beta_k$.

{\em Case 2:} the rule decomposes an $\SL$ inside $\Gamma$ (i.e., in $\alpha_i$), and $\beta_k \BS b$ goes not to the main premise:
$$
\infer{\alpha_1, \ldots, \alpha'_i \SL \alpha''_i, \alpha_{i+1}, \ldots, \alpha_\ell, \beta_1 \BS b, \ldots, \beta_{k-1} \BS b,
\beta_k \BS b \seqarr b}
{\alpha_{i+1}, \ldots, \alpha_\ell, \beta_1 \BS b, \ldots, \beta_{k-1} \BS b,
\beta_k \BS b \seqarr \alpha''_i & 
\alpha_1, \ldots, \alpha'_i \seqarr b}
$$
The main premise fails to be derivable, due to Lemma~\ref{Lm:aux1}. Contradiction.

{\em Case 3:} the rule decomposes $\BS$ in the rightmost $\beta_k \BS b$; $\Gamma = \Gamma_1, \Gamma_2$.
$$
\infer{\Gamma_1, \Gamma_2, \beta_k \BS b \seqarr b}
{\Gamma_2 \seqarr \beta_k & \Gamma_1, b \seqarr b}
$$
Here $\Gamma_1, b \seqarr$ is a sequent of the form suitable for Lemma~\ref{Lm:aux2}, and by this lemma $\Gamma_1$ should be empty. Thus, $\Gamma = \Gamma_2$, and $\Gamma \seqarr \beta_k$ is derivable.
\end{proof}

Now we are ready to prove the pseudo-double-negation theorem (Theorem~\ref{Th:DN}). Recall its formulation:

\begin{theorem*}
 If $b$ does not occur in $\alpha_1, \ldots, \alpha_n, \beta$, then $\alpha_1, \ldots, \alpha_n \vdash \beta$ is derivable in $\MALC$ if and only if so is 
 $\alpha_1^{bb}, \ldots, \alpha_n^{bb} \vdash \beta^{bb}$.
\end{theorem*}

\begin{proof}
 We use induction on $k$ to prove the following statement: $\alpha_1, \ldots, \alpha_n \seqarr \beta$ is
 equiderivable with 
 $\alpha_{k+1}, \ldots, \alpha_n, \beta^b, \alpha_1^{bb},\ldots, \alpha_k^{bb} \seqarr b$. In particular, for $k=n$ we shall have $\beta^b, \alpha_1^{bb}, \ldots, \alpha_n^{bb} \seqarr b$, which is equiderivable with
 $\alpha_1^{bb}, \ldots, \alpha_n^{bb} \seqarr \beta^{bb}$ by
 $(\seqarr\BS)$ and $(\seqarr\BS)_\inv$.
 
 First let us prove induction base ($k = 0$).
 Sequents $\alpha_1, \ldots, \alpha_n, \beta \BS b \seqarr b$ and  $\alpha_1, \ldots, \alpha_n \seqarr \beta$ are equiderivable by $(\BS \seqarr)$ and Lemma~\ref{Lm:aux3}. 
 
 For the induction step, we have the following.
 \begin{enumerate}
  \item The $(k+1)$-st sequent is derived from the $k$-th one
  as follows:
  $$
  \infer{\alpha_{k+2}, \ldots, \alpha_n, \beta^b, \alpha_1^{bb}, \ldots, \alpha_k^{bb}, (\alpha_{k+1} \BS b) \BS b \seqarr b}
  {\infer{\alpha_{k+2}, \ldots, \alpha_n, \beta^b, \alpha_1^{bb}, \ldots, \alpha_k^{bb} \seqarr
  \alpha_{k+1} \BS b}
  {\alpha_{k+1}, \alpha_{k+2}, \ldots, \alpha_n, \beta^b, \alpha_1^{bb}, \ldots, \alpha_k^{bb} \seqarr b}
  & b \seqarr b}
    $$
    \item {\em Vice versa,} if  $\alpha_{k+2}, \ldots, \alpha_n, \beta^b, \alpha_1^{bb}, \ldots, \alpha_k^{bb}, (\alpha_{k+1} \BS b) \BS b \seqarr b$ is derivable, then by Lemma~\ref{Lm:aux3} so is
    $\alpha_{k+2}, \ldots, \alpha_n, \beta^b, \alpha_1^{bb}, \ldots, \alpha_k^{bb} \seqarr \alpha_{k+1} \BS b$. By \mbox{$(\seqarr\BS)_\inv$} we get derivability of 
    $\alpha_{k+1}, \alpha_{k+2}, \ldots,  \alpha_n, \beta^b, \alpha_1^{bb}, \ldots, \alpha_k^{bb} \seqarr b$.
 \end{enumerate}

 Thus the $(k+1)$-st sequent is equiderivable with the $k$-th one, and the latter, by induction hypothesis, is equiderivable with $\alpha_1, \ldots, \alpha_n \seqarr \beta$.
\end{proof}

\end{document}